\documentclass[pdflatex,sn-basic]{sn-jnl}

\usepackage{graphicx}%
\usepackage{multirow}%
\usepackage{amsmath,amssymb,amsfonts}%
\usepackage{amsthm}%
\usepackage{mathrsfs}%
\usepackage[title]{appendix}%
\usepackage{xcolor}%
\usepackage{textcomp}%
\usepackage{manyfoot}%
\usepackage{booktabs}%
\usepackage{algpseudocode}%
\usepackage{listings}%

\usepackage{tabularx}
\usepackage{siunitx}
\usepackage{bbm}

\usepackage[ruled, nofillcomment]{algorithm2e}
\SetKw{KwAnd}{and}
\SetKwComment{Comment}{/* }{ */}
\SetKw{Initialization}{Initialization}
\SetKw{Data}{Data}

\usepackage{subcaption}

\theoremstyle{thmstyleone}%
\newtheorem{theorem}{Theorem}
\theoremstyle{thmstyletwo}%

\theoremstyle{thmstylethree}%

\allowdisplaybreaks

\newcommand{\diff}{\,\mathrm{d}}

\DeclareMathOperator*{\argmax}{\arg\!\max}

\newcommand{\figref}[1]{Figure~\ref{#1}}

\begin{document}

\title[Efficient Inference in First Passage Time Models]{Efficient Inference in First Passage Time Models}


\author*[1]{\fnm{Sicheng} \sur{Liu}}\email{sicheng\_liu@brown.edu}

\author[2]{\fnm{Alexander} \sur{Fengler}}\email{alexander\_fengler@brown.edu}

\author[2,3]{\fnm{Michael J.} \sur{Frank}}\email{michael\_frank@brown.edu}

\author[1]{\fnm{Matthew T.} \sur{Harrison}}\email{matthew\_harrison@brown.edu}

\affil*[1]{\orgdiv{Division of Applied Mathematics}, \orgname{Brown University}, \orgaddress{\street{182 George St}, \city{Providence}, \postcode{02912}, \state{RI}, \country{USA}}}

\affil[2]{\orgdiv{Department of Cognitive and Psychological Sciences}, \orgname{Brown University}, \orgaddress{\street{190 Thayer St}, \city{Providence}, \postcode{02912}, \state{RI }, \country{USA}}}

\affil[3]{\orgdiv{Carney Institute for Brain Science}, \orgname{Brown University}, \orgaddress{\street{164 Angell St}, \city{Providence}, \postcode{02912}, \state{RI}, \country{USA}}}


\abstract{First passage time models describe the time it takes for a random process to exit a region of interest and are widely used across various scientific fields. Fast and accurate numerical methods for computing the likelihood function in these models are essential for efficient statistical inference of model parameters. Specifically, in computational cognitive neuroscience, generalized drift diffusion models (GDDMs) are an important class of first passage time models that describe the latent psychological processes underlying simple decision-making scenarios. GDDMs model the joint distribution over choices and response times as the first hitting time of a one-dimensional stochastic differential equation (SDE) to possibly time-varying upper and lower boundaries. They are widely applied to extract parameters associated with distinct cognitive and neural mechanisms. However, current likelihood computation methods struggle in common application scenarios in which drift rates dynamically vary within trials as a function of exogenous covariates (e.g., brain activity in specific regions or visual fixations). In this work, we propose a fast and flexible algorithm for computing the likelihood function of GDDMs based on a large class of SDEs satisfying the Cherkasov condition. Our method divides each trial into discrete stages, employs fast analytical results to compute stage-wise densities, and integrates these to compute the overall trial-wise likelihood. Numerical examples demonstrate that our method not only yields accurate likelihood evaluations for efficient statistical inference, but also considerably outperforms existing approaches in terms of speed.}




\maketitle

\section{Introduction}\label{sec: intro}
First passage time models (also known as first hitting time models) are a fundamental class of stochastic process models that describe the time it takes for a random process to exit a region of interest. These models are widely used in various fields, including physics, finance, biology, and cognitive science, to analyze phenomena such as the time until a stock hits a particular price level, the time it takes for a neuron to fire, a survival time of a transplant patient, or the time to reach a decision barrier that triggers a choice in psychological experiments \citep{bhattacharya2021random, metzler2014first, gardiner2009stochastic, lee2006threshold}.
Accurately computing the density of the first passage time is crucial for the practical application of these models, such as in scientific hypothesis testing and parameter estimation. Since analytical solutions for many complex real-world first passage time models are unavailable, effective numerical methods are an indispensable part of the computational toolkit for applying these models. In this paper, we introduce a fast and flexible method for computing the first passage time density in a large class of models arising in cognitive neuroscience, mathematical psychology, and economics.

Consider a psychological experiment in which a subject must choose between two alternatives, A and B. The experimenter records both the choice and the time taken to make it. A common approach to modeling this behavior is through a stochastic process $X(t)$ that starts at zero and evolves randomly over time until it first reaches ${+}a$ or ${-}a$, where $a$ is the evidence threshold. The choice (A vs. B) is modeled by which boundary is hit (${+}a$ vs.~${-}a$) and the response time is modeled by the first hitting time (Fig. \ref{fig: simple-ddm}). 

\begin{figure*}[htb!]
	\centering
	\includegraphics[width=\textwidth]{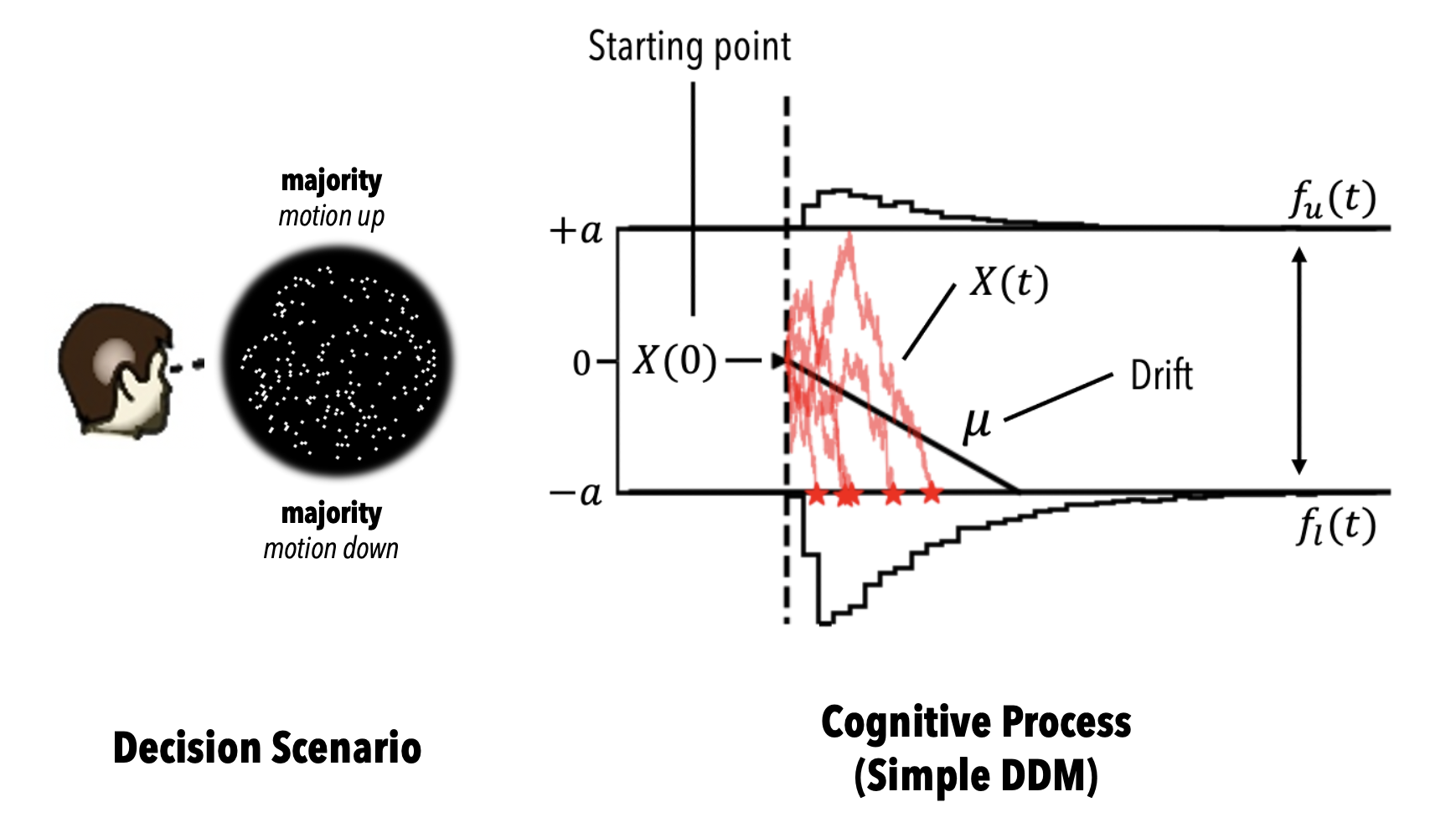}
	\label{fig: single-stage-basic}
	\caption{The simple DDM in the context of a canonical experiment in cognitive science, the \textit{dot-motion task}. The DDM serves to represent the \textit{cognitive process} underpinning the choices in this experiment. The task explicitly given to participants: Decide if the majority of dots on the screen are moving up or down. We show trajectories of the corresponding SDE, $\diff X(t) = \mu  \diff t + \diff W(t)$, in \textcolor{red}{red}, and example decision points as \textit{red stars}. The first hitting times distributions $f_u(t), \ f_\ell(t)$ for the corresponding choices (\textit{up}, \textit{down}) are shown above the decision boundaries $\{+a, -a\}$.}
	\label{fig: simple-ddm}
\end{figure*}

A widely used model for $X(t)$ is a Brownian motion with drift
\[ X(t) = \mu t + W(t) , \]
where $W(t)$ is a standard Brownian motion. The drift rate $\mu$ quantifies the strength of the evidence that favors A over B and affects the joint distribution of the choice and the response time. Models such as this are often called {\bf drift diffusion models (DDMs)} \citep{ratcliff2008, fengler2022} and written in differential form, such as,
\begin{equation}\label{eqn: sde0}
	\diff X(t) = \mu  \diff t + \diff W(t)  
\end{equation}
with the initial condition $X(0)=0$.  

This basic DDM can be modified in a variety of ways to better model experimental data and test theories about the biological processes underlying decision making. For instance, the boundaries ($\pm a$) can be asymmetric and vary over time to reflect changing pressures to make a decision; the starting position $X(0)$ can be nonzero or random, representing an initial decision bias; the drift rate towards the boundaries can depend on the presented options and vary over time; and the random process $X(t)$ can have more complicated stochastic dynamics than Brownian motion \citep{fengler2022, usher2001time, malhotra2018, palestro2018, evans2020parameter}. We adopt the terminology from \cite{shinn2020flexible} and call these models generalized drift diffusion models (GDDMs). Simple, computationally tractable versions of GDDMs have found widespread use in cognitive neuroscience and mathematical psychology \citep{ratcliff1978, ratcliff2008, frank2015, doi2020, yartsev2018}. While more complex versions of GDDMs have been proposed to encode a broader range of scientific hypotheses, their quantitative assessment and applications are currently limited due to computational challenges in parameter estimation \citep{fengler2022}. In particular there do not exist fast algorithms for calculating the likelihood function of the parameters (i.e., the probability density of the hitting time on the boundary reached) in some of the more complex GDDMs that are currently of scientific interest \citep{krajbich2010visual, GingJehli2025, pagnier2024double, kuhn2025computational}.

Computational speed is crucial in many cases for the following reason: As noted in \cite{lee2006threshold}, first passage time models commonly incorporate covariates that may vary across individuals and even across individual trials. This is particularly relevant in computational cognitive neuroscience, where the precise specification of the GDDM might change every time a subject makes a choice, i.e., in each experimental trial. An experiment with, say, $50$ subjects, each completing $200$ trials, might require evaluating the likelihoods of $10^4$ different versions of the GDDM to get the overall dataset likelihood. For instance, in equation \eqref{eqn: sde0} the drift rate $\mu$ might be different on every trial if the choices presented to subjects are different on every trial. The problem becomes even more difficult when the trial-specific drift rate can vary dynamically within a trial. 

For a specific example, our original motivation to study this problem stems from  attentional drift diffusion models (aDDMs)\citep{krajbich2010visual}, which have been proposed as a generative model for investigating how visual attention influences decision making. In aDDMs, the drift rate $\mu(t)$ in each trial is allowed to covary on a moment-by-moment basis, depending on the subject's eye movements toward one choice option or another, which are also recorded in the experimental setup. The entire shape of $\mu(t)$ changes dramatically on every trial. This model is a generic example of the general class of GDDMs where instantaneous within-trial dynamics can be based on exogenous covariates such as brain activity or eye gaze position, hence complicating the numerical procedures for computing the first passage time densities and further applications. On top of this, statistical inference tasks would still require repeated evaluations of the total dataset likelihood at different parameter values, both in frequentist and Bayesian frameworks. In particular, hierarchical Bayesian inference is widely employed in cognitive science to estimate both group- and subject-level parameters simultaneously \citep{wiecki2013hddm, fengler2022}. The heavy reliance on iterative Markov chain Monte Carlo (MCMC) methods in this setting further intensifies the need for fast likelihood evaluations.

In this paper we introduce fast algorithms for accurately approximating the probability density of the hitting time in a class of GDDMs that have continuous lower and upper boundaries, $\ell(t)$ and $u(t)$, and with dynamics that satisfy the It\^{o}'s stochastic differential equation (SDE)
\begin{equation}\label{eqn: sde}
	\diff X(t)=\mu(X(t), t) \diff t+\sigma(X(t), t) \diff W(t)
\end{equation}
for possibly random $X(0)$ and for functions $\mu$ and $\sigma$ satisfying the so-called Cherkasov condition \citep{cherkasov1957transformation, ricciardi1976transformation}. Loosely speaking, the Cherkasov condition means that the complicated SDE in equation \eqref{eqn: sde} can be transformed into a simpler SDE more akin to equation \eqref{eqn: sde0}. A special case that includes many cases of practical interest, such as the aDDM, is when $\mu$ and $\sigma$ depend only on $t$ and not on $X(t)$\footnote{In the cognitive science context, the terminology ``DDM" is typically adopted to denote these models that directly integrate noisy accumulating evidence. The more general model in equation \eqref{eqn: sde}, including some examples shown later in this paper (e.g., the Ornstein-Uhlenbeck model), fall within the GDDM framework.}. There exist numerical approaches and software for approximating hitting time densities in models even more general than equation \eqref{eqn: sde}, but in many cases they are currently too slow for practical use in real-world psychological experiments.

Our approach shares some conceptual similarities to a line of previous work \citep{wang1997boundary, novikov1999approximations, potzelberger2001boundary, wang2007crossing, jin2017first}. Both those methods and ours transform the problem into a canonical representation, subdivide time into discrete stages that allow for analytical approximations, compute the final approximation through integration over these stages, and then map the result back to the original problem. However, none of these earlier approaches are directly applicable to our setting or capable of statistical computation at the scale of the present work (see Section \ref{sec: related_work}). In contrast, by leveraging accurate analytical results and numerical integration in a sequential manner, our semi-analytical-numerical approach achieves substantially higher computational efficiency. As shown in Section \ref{sec: timing}, it only requires microseconds for our method to evaluate a single likelihood value in a realistic aDDM example. These innovations and efficiency gains are crucial for practical likelihood-based inference in modern experimental settings. 


The remainder of this article is structured as follows. Section \ref{sec: ddm} begins with a mathematical description of the GDDM. Section \ref{sec: related_work} summarizes previous work related to this paper, including other efforts to address the computational challenges in first passage time models. In Section \ref{sec: methods}, we introduce our proposed method for computing the first passage time density in certain GDDMs, starting with analytical results for simple DDMs, then combining them sequentially based on Markovian properties. We also apply variable transformations and piecewise linear approximations to extend the method to more complex GDDMs, including those that allow for position-dependent drift and diffusion terms as well as dynamic decision boundaries \citep{usher2001time,fengler2022,GingJehli2025}. 
Section \ref{sec: examples} shows several numerical examples that demonstrate the speed and accuracy of our method for both likelihood computation and statistical inference tasks. Section \ref{sec: conclusion} summarizes the present work, including its applicability and suggestions for future work. It is worth mentioning that while our examples are drawn from cognitive neuroscience and mathematical psychology, the method we propose is equally applicable to first passage time models in other fields. For instance, it could be extended to financial applications such as pricing double barrier options \citep{lin1998double}, as well as other areas where first passage times play a critical role.

\section{Drift Diffusion Model}\label{sec: ddm}
Here we present a mathematical description of GDDMs. Let $W(t)$ be a one-dimensional standard Brownian motion and let $X(t)$ be a (weak) solution of the It\^o's stochastic differential equation (SDE) in equation \eqref{eqn: sde} associated with $W(t)$, where $X(0)$ is independent of $W$. Let $u(t)$ and $\ell(t)$ be two continuous functions representing the upper and lower boundaries, respectively, satisfying $\mathbb{P}(\ell(0)<X(0)<u(0))=1$ and $\ell(t) \leq u(t)$ for all $t > 0$.

In this setting, we aim to study the joint distribution of the first exit time of $X(t)$ from the time-dependent region $(\ell(t), u(t))$ and the specific boundary exited (i.e., which choice is made). The first exit time is denoted as $\tau = \inf \{t>0: X(t) \notin(\ell(t), u(t))\}$, and the boundary hit is denoted as $C$, which takes values in the symbol set $\{u, \ell\}$ to represent which boundary is hit. The event of $X(t)$ hitting the boundary $C$ at time $\tau$ models the participant's decision to choose option $C$ at time $\tau$. The event $\{\tau=\infty\}$ can never be observed in practice, so we do not explicitly discuss it below. If $\mathbb{P}(\tau = \infty) > 0$, then the probability density functions of $\tau$ discussed in this paper are technically sub-probability density functions that integrate to $\mathbb{P}(\tau < \infty) < 1$. This does not affect any of the calculations needed for likelihood-based inference.

We denote the sub-probability density functions 
\begin{equation*}
	\begin{aligned}
		f_u(t)&=\frac{\diff}{\diff t}\mathbb{P}(\tau\le t, C=u)\\
		f_{\ell}(t)&=\frac{\diff}{\diff t}\mathbb{P}(\tau\le t, C=\ell)    
	\end{aligned}
\end{equation*}
as the first passage time densities (FPTDs) on the upper boundary and lower boundary\footnote{Note that ``first passage time density on the upper (lower) boundary" is sometimes interpreted as the density of $\tau_u=\inf\{t>0:X(t)\ge u(t)\}(\tau_\ell=\inf\{t>0:X(t)\le \ell(t)\})$ when only considering one upper boundary $u(t)$ (lower boundary $\ell(t)$). In this case $\tau=\tau_u\wedge\tau_\ell$. These quantities represent different aspects compared to our definition of the FPTD on the upper and lower boundaries and should not be confused.}, respectively, implicitly assuming throughout sufficient regularity conditions on the model for these densities to exist (see the sufficient conditions for density existence in \cite{ferebee1982tangent, strassen1967almost}). The marginal density of $\tau$ is thus given by $\frac{\diff}{\diff t}\mathbb{P}(\tau\le t) = f_u(t)+f_\ell(t)$. With this notation, we can write the joint probability density of $(\tau, C)$ as 
\begin{equation} \label{eqn: f(t,c)}
	f(t, c)
	= f_u(t)\mathbbm{1}_{\{u\}}(c)+f_{\ell}(t)\mathbbm{1}_{\{\ell\}}(c) ,
\end{equation}
where $\mathbbm{1}_E(\omega)$ is the indicator function of $E$ that is $1$ if $\omega\in E$ and $0$ otherwise.

If we consider the possible truncation of the process at $t=T < \infty$, then the first exit from the open interval $(\ell(t), u(t))$ occurs either through the upper boundary $u(t)$ before $T$ (i.e., the event $\{\tau \le T, C=u\}$), through the lower boundary $\ell(t)$ before $T$ (i.e., the event $\{\tau \le T, C=\ell\}$), or does not occur before $T$ (i.e., the event $\{\tau > T\}$), in which case the process remains inside $(\ell(t), u(t))$ up to $T$ and reaches the truncation boundary $t=T$. These three mutually exclusive outcomes give rise to three sub-probability densities:
\begin{itemize}
	\item $f_u(t)\mathbbm{1}_{(0, T]}(t)$: the FPTD on the upper boundary $u(t)$, truncated at $t=T$.
	\item $f_\ell(t)\mathbbm{1}_{(0, T]}(t)$: the FPTD on the lower boundary $\ell(t)$, truncated at $t=T$.
	\item $q(x)= \frac{\diff}{\diff x} \mathbb{P}(X(T)\le x, \tau > T)$: the density of the process's position at time $T$, provided that it has not hit either upper or lower boundary yet. We refer to $q$ as the non-passage density (NPD) at $T$.
\end{itemize}
By the conservation of total probability we always have
\begin{equation*}
	\int_0^T f_u(t)\diff t + \int_0^T f_\ell(t)\diff t + \int_{\ell(T)}^{u(T)} q(x)\diff x=1
\end{equation*}
for any $T>0$. Define the non-passage probability (NPP) at $T$ to be
\begin{equation} \label{eqn: Q} Q = \mathbb{P}(\tau > T) = \int_{\ell(T)}^{u(T)} q(x)\diff x . \end{equation} 
$Q$ is useful in settings where the process is artificially stopped whenever $\tau > T$. For instance, in some psychological experiments, subjects may be required to respond within a certain amount of time $T$, otherwise the trial ends with the outcome ``non-response'' and the NPP $Q$ is the probability of non-response. 

The goal of this paper is to obtain a fast and accurate numerical approximation of the joint density $f(t,c)$ in equation \eqref{eqn: f(t,c)} and the NPP $Q$ in equation \eqref{eqn: Q}. To this end, we develop fast numerical approximations of FPTDs and NPDs, the latter of which are needed for intermediate calculations in our approach.

As mentioned in the Introduction, we are motivated by practical problems for which we need to evaluate many {\em different} joint densities \eqref{eqn: f(t,c)}. We define $\boldsymbol{\Xi}=\{\mu(\cdot), \sigma(\cdot), u(\cdot), \ell(\cdot), q_0(\cdot)\}$ to be the configuration of a GDDM, where $q_0$ is a representation for the distribution of $X(0)$, usually a probability density function or a point mass at a deterministic starting point $x_0$. In statistical applications, $\boldsymbol{\Xi}$ will depend on unknown parameters and known covariates, some of which might vary within and across each trial. If the $i$-th trial results in the observation $(\tau_i,C_i)$, then the contribution to the likelihood for that trial is $f(\tau_i,C_i\mid\boldsymbol{\Xi}_i)$. This is simply the FPTD on the boundary $C_i$ evaluated at time $\tau_i$ using the GDDM configuration $\boldsymbol{\Xi}_i$ for the $i$-th trial. In experiments with a time limit, some trials might result in non-response, in which case the contribution to the likelihood is the NPP for that GDDM configuration, say $Q_i=Q(\boldsymbol{\Xi}_i)$. Consequently, denoting $R_i=\mathbbm{1}\{\text{a response on trial $i$}\}$, the likelihood function is then given by\footnote{The statistical model for $\tau$ here is a mixed distribution consisting of a continuous part supported on $\tau \in(0, T]$ (response) and a point mass at $\tau>T$ (no response). Mixed likelihoods of this form are common in survival analysis, for example see the Section 2.2 of \cite{lawless2011statistical}.}
\begin{equation*}
\begin{aligned}
&\mathcal{L}_{\text{GDDM}}(\boldsymbol{\theta}\mid ((\tau_i, C_i))_{i=1}^{n})=\prod_{i=1}^{n} f(\tau_i, C_i \mid \boldsymbol{\Xi}_i)^{R_i}Q(\boldsymbol{\Xi}_i)^{1-R_i}.
\end{aligned}
\end{equation*}
A usual experiment in the cognitive sciences, where GDDMs are ubiquitously applied, produces a dataset with thousands or tens of thousands of trials (across multiple subjects). Consequently, fast and accurate methods for evaluating FPTDs and NPPs are crucial for the practical use of likelihood-based inference in these models.

\section{Related Work}\label{sec: related_work}
We briefly review some commonly used computational methods for calculating the density or distribution of first passage times in GDDMs.

\textbf{Analytic or semi-analytic methods.} For certain simple DDMs, analytic formulas for the first passage time densities or distributions can be derived, often in the form of infinite series. These methods are fast and accurate, but are typically applicable only for very specific cases. For example, they are available for the standard DDM where all parameters $\mu, \sigma, x_0, u, \ell$ are constant \citep{navarro2009fast, blurton2012fast, gondan2014even}. More recently, \cite{srivastava2017martingale} used martingale theory to derive closed-form formulas that extend certain theoretical analyses to GDDMs.

The approach in this paper is based on numerically extending these fast special cases to much broader classes of GDDMs. Along this line, earlier works \citep{wang1997boundary, novikov1999approximations, potzelberger2001boundary, wang2007crossing} computed NPPs,
while \cite{jin2017first}, \cite{ichiba2011efficient}, and the present work target FPTDs,
although \cite{jin2017first} and \cite{ichiba2011efficient} restrict attention to the special case of a one-sided boundary. The FPTD case is more challenging but also more useful for statistical inference since it directly gives the likelihood. A key limitation of many prior approaches \citep{wang1997boundary, potzelberger2001boundary, wang2007crossing, ichiba2011efficient, jin2017first}, is that they rely on Monte Carlo integration, 
which is infeasible for large-scale real-world applications.

\textbf{Numerical solutions of governing equations.} These methods exploit the known connections between diffusion processes and partial differential equations (PDEs), called Kolmogorov equations (see the Chapter 8 of \cite{weinan2021applied} for an introduction, and Chapter 5.7 of \cite{karatzas2014brownian} for a mathematically rigorous treatment). Both the Kolmogorov forward equation (KFE, also known as Fokker-Planck equation in physics) \citep{shinn2020flexible, diederich2003simple} and the Kolmogorov backward equation (KBE) \citep{voss2008fast, voss2007fast} have been exploited to derive the FPTDs. Moreover, the first passage time problem can be formulated into Volterra integral equations of renewal type \citep{smith2000stochastic, peskir2002integral}, where applications to DDMs \citep{smith2022modeling, paninski2008integral} have also been explored. In these works, numerical methods, particularly finite difference methods \citep{thomas2013numerical}, have been developed to solve these equations, enabling the efficient computation of FPTDs in certain cases. Other numerical methods like discontinuous Galerkin methods have also been considered, see \cite{patie2008first} and references therein. However, these approaches remain computationally expensive, especially in settings where the GDDM configuration differs across trials, including but not limited to cases with trial-specific covariates that vary dynamically over time. For a comprehensive review of these methods, we refer readers to \cite{richter2023diffusion}.

\textbf{Stochastic simulation.} Methods that can sample from the hitting time distribution of the target GDDM can be used to approximate FPTDs via density estimation. Efficient methods for exact sampling are challenging; see \cite{herrmann2020exact, herrmann2022exact} and the references therein. If the goal is statistical inference, then these sampling algorithms can be incorporated into a variety of simulation-based inference (SBI) procedures; see \cite{cranmer2020frontier} for review and \cite{fengler2021likelihood} for an SBI example with DDMs. Nevertheless, the Monte Carlo methods underlying SBI are computationally intensive and can become impractical for large datasets, particularly when $\Xi$ involves high-dimensional parameters or covariates that vary dynamically over time and across trials.

\section{Main Methods}\label{sec: methods}
In this section, we present our main algorithm for computing FPTDs and NPDs. We begin by studying a simple case where analytical results are known in Section \ref{sec: single-stage} then generalize it to more complicated cases in Sections \ref{sec: int single-stage}--\ref{sec: generalization}.

\subsection{Single-stage Model}\label{sec: single-stage}
We consider the Brownian motion with drift
\begin{equation}\label{eqn: single-stage}
	\diff X(t)=\mu \diff t+\sigma\diff W(t)
\end{equation}
starting at $x_0$. There are two linear absorbing boundaries
\begin{equation*}
	\begin{aligned}
		\text{upper: }u(t)=a_1+b_1t,\quad \text{lower: }\ell(t)=a_2+b_2t
	\end{aligned}
\end{equation*}
with $a_2<x_0<a_1$, and, additionally, a truncation boundary at $t=T$
such that $\ell(t) < u(t)$ on $[0,T)$. We name this the ``\textbf{single-stage model}" as it is the building block of our method.

\cite{hall1997distribution} derives the following analytical formulas\footnote{Another work along this direction is \cite{anderson1960modification}. We adopt the results in \cite{hall1997distribution} since they are simpler and were shown to agree with those of \cite{anderson1960modification}. Equation \eqref{eqn: fptd_upper}, \eqref{eqn: fptd_lower}, and \eqref{eqn: fptd_vertical} have been slightly rewritten for use in this paper, compared to their original form in \cite{hall1997distribution}.} for FPTDs and NPDs in the single-stage model for the special case when $x_0=0$ and $\sigma=1$: We define
$\overline{a}=\frac{a_1+a_2}{2},\ \overline{b}=\frac{b_1+b_2}{2},\ c=a_1-a_2,\ b=\frac{b_2-b_1}{2}$, then
\begin{itemize}
	\item The FPTD on the upper boundary $u(t)$ is given by
	\begin{equation}\label{eqn: fptd_upper-basic}
		\begin{aligned}
			&f^\text{basic}_u(t;\mu,a_1, b_1, a_2, b_2, T) = \\
			&\quad\frac{1}{\sqrt{2\pi t^3}}  e^{-\frac{b}{c} a_1^2+a_1 \delta_u-\frac{1}{2} \delta_u^2 t}\sum_{j=0}^{\infty}(-1)^j \alpha_j e^{\left(\frac{b}{c}-\frac{1}{2t}\right)\alpha_j^2}\mathbbm{1}_{(0, T]}(t)
		\end{aligned}
	\end{equation}
	where $\delta_u=\mu-b_1,\alpha_j=(j+\frac{1}{2})c+(-1)^j\overline{a}$.
	\item Similarly, the FPTD on the lower boundary $\ell(t)$ is given by
	\begin{equation}\label{eqn: fptd_lower-basic}
		\begin{aligned}
			&f^\text{basic}_\ell(t;\mu,a_1, b_1, a_2, b_2, T)=\\
			&\quad\frac{1}{\sqrt{2\pi t^3}}  e^{-\frac{b}{c} a_2^2-a_2 \delta_\ell-\frac{1}{2} \delta_\ell^2 t}\sum_{j=0}^{\infty}(-1)^j \beta_j e^{\left(\frac{b}{c}-\frac{1}{2t}\right)\beta_j^2}\mathbbm{1}_{(0, T]}(t)
		\end{aligned}
	\end{equation}
	where $\delta_\ell=-\mu+b_2, \beta_j=(j+\frac{1}{2})c-(-1)^j\overline{a}$.
	
	\item Letting $y=x-\overline{b}T$, the NPD at the vertical boundary $t=T$ is given by
	\begin{equation}\label{eqn: fptd_vertical-basic}
		\begin{aligned}
			&q^\text{basic}(x;\mu,a_1, b_1, a_2, b_2, T)=\\
			&\quad\frac{1}{\sqrt{2\pi T}}e^{(\mu-\overline{b}) x-\frac{1}{2}(\mu^2-\overline{b}^2) T}\bigg\{ e^{-\frac{y^2}{2T}}+\\
            &\quad\sum_{j=1}^{\infty}\bigg[e^{4 b j(j c-\overline{a})-\frac{(y-2 j c)^2}{2T}}-e^{2 b(2 j-1)\left(j c-a_1\right)-\frac{(y+2 j c-2 a_1)^2}{2T}}+\\
            &\quad e^{4 b j(j c+\overline{a})-\frac{(y+2 j c)^2}{2T}}-e^{2 b(2 j-1)\left(j c+a_2\right)-\frac{(y-2 j c-2 a_2)^2}{2T}}\bigg]\bigg\}\mathbbm{1}_{(\ell(T), u(T))}(x).
		\end{aligned}
	\end{equation}
\end{itemize}
For the single-stage model \eqref{eqn: single-stage} with a nonzero starting point $x_0$ and a diffusion coefficient $\sigma$ that does not equal $1$, it is straightforward to obtain its FPTDs and NPD by some shifting and scaling: For notational convenience we further let $\mathcal{B}=(a_1, b_1, a_2, b_2, T)$ denote the collection of all the boundary parameters, then we have
\begin{align} \label{eqn: fptd_upper}
	& f_u^{\text{single}}(t ; \mu, \sigma, \mathcal{B}, x_0)= f_u^{\text{basic}}\Big(t ; \tfrac{\mu}{\sigma}, \tfrac{a_1-x_0}{\sigma}, \tfrac{b_1}{\sigma}, \tfrac{a_2-x_0}{\sigma}, \tfrac{b_2}{\sigma}, T\Big), \\ \label{eqn: fptd_lower}
	& f_{\ell}^{\text{single}}(t ; \mu, \sigma, \mathcal{B}, x_0)=  f_{\ell}^{\text{basic}}\Big(t ; \tfrac{\mu}{\sigma}, \tfrac{a_1-x_0}{\sigma}, \tfrac{b_1}{\sigma}, \tfrac{a_2-x_0}{\sigma}, \tfrac{b_2}{\sigma}, T\Big),\\ \label{eqn: fptd_vertical}
	& q^{\text{single}}(x ; \mu, \sigma, \mathcal{B}, x_0)=\tfrac{1}{\sigma} q^{\text{basic}}\Big(\tfrac{x-x_0}{\sigma} ; \tfrac{\mu}{\sigma}, \tfrac{a_1-x_0}{\sigma}, \tfrac{b_1}{\sigma}, \tfrac{a_2-x_0}{\sigma}, \tfrac{b_2}{\sigma}, T\Big).
\end{align}
An illustration of the single-stage model is shown in \figref{fig: single-stage}(a).

\begin{figure*}[htb!]
	\centering
	\includegraphics[width=0.8\textwidth]{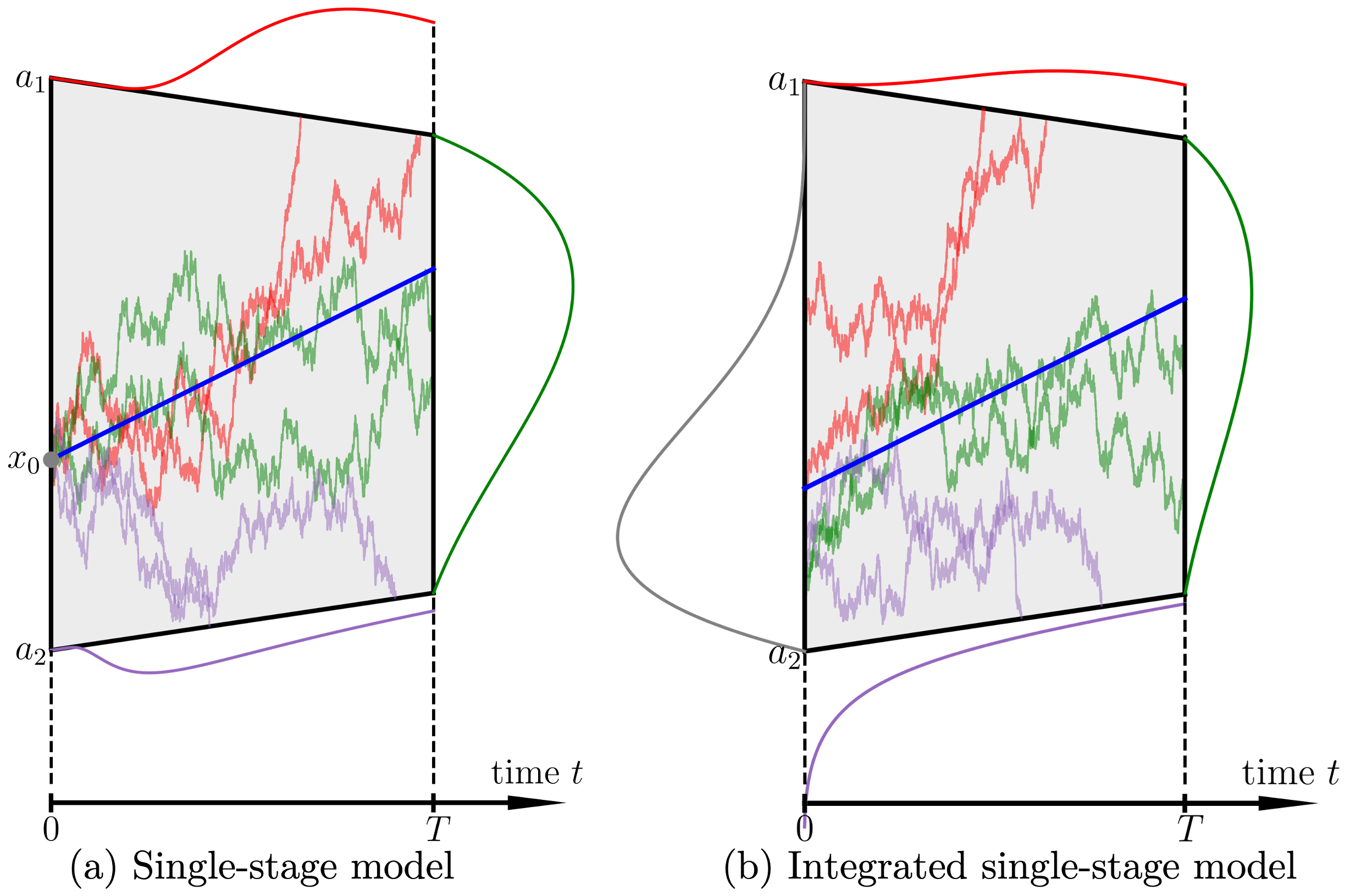}
	\label{fig: single-stage1}
	\caption{(a) A single-stage model is a GDDM with constant drift $\mu$ and two linear boundaries $u(t), \ell(t)$ starting at $x=x_0$ and truncated at $t=T$. The process can hit one of $u(t)$ and $\ell(t)$ before $T$, or hit neither. This results in three distinct sub-probability distributions, corresponding to these three outcomes, with their sub-probability densities plotted by red, purple, and green curves along the boundaries. Sample trajectories, colored to match their respective outcomes, are also displayed. The blue line represents the expected value of the process, given by $\mathbb{E}[X(t)]=x_0+\mu t$. In a single-stage model the total area under the three sub-probability density curves sums to 1.  (b) The integrated model extends the single-stage model by introducing random variability in the initial position $X(0)$. The expected value of the process in this model is given by $\mathbb{E}[X(t)]=\mathbb{E}[X(0)]+\mu t$, represented by the blue line. The density of $X(0)$, denoted by $q_0$, is illustrated by the grey curve. We allow $q_0$ to be a sub-probability density function. In an integrated single-stage model, the total area under the three outcome sub-probability density curves sums to the area under the initial grey curve.
	}
	\label{fig: single-stage}
\end{figure*}

\subsection{Integrated Single-stage Model}\label{sec: int single-stage}

The single-stage model uses $X(0)=x_0$ deterministically. If, instead, $X(0)$ is random, then the corresponding FPTDs and NPDs can be obtained by taking the expectations of 
\eqref{eqn: fptd_upper}, \eqref{eqn: fptd_lower}, and \eqref{eqn: fptd_vertical} over the distribution of $X(0)$. We call this the {\bf integrated single-stage model}. If the distribution of $X(0)$ has a probability density function $q_0$ supported on $(a_2,a_1)$, then FPTDs and NPDs for the integrated single-stage model are
\begin{equation} \label{eqn: int}
	\begin{aligned}
		&f_u^{\text{int}}(t;\mu,\sigma, \mathcal{B},q_0)=\int_{a_2}^{a_1} f_u^{\text{single}}(t;\mu,\sigma, \mathcal{B}, x_0)q_0(x_0)\diff x_0 , \\
		&f_\ell^{\text{int}}(t;\mu, \sigma,\mathcal{B},q_0)=\int_{a_2}^{a_1} f_\ell^{\text{single}}(t;\mu, \sigma,\mathcal{B}, x_0)q_0(x_0)\diff x_0 , \\
		&q^{\text{int}}(x;\mu, \sigma, \mathcal{B},q_0)=\int_{a_2}^{a_1} q^{\text{single}}(x;\mu, \sigma, \mathcal{B}, x_0)q_0(x_0)\diff x_0 . 
	\end{aligned}
\end{equation}
An illustration of the integrated single-stage model when $X(0)$ has a probability density function is shown in \figref{fig: single-stage}(b). When $X(0)$ is a discrete random variable, we abuse notation and allow $q_0=\sum_{j=1}^J w_{0j} \delta_{x_{0j}}$ to be a weighted sum of point masses, where $(w_{01}, \ldots, w_{0J})$ is a probability mass function on $(x_{01}, \ldots, x_{0J})$ and $\delta$ is the Dirac delta generalized function, so that equation \eqref{eqn: int} still holds formally. For instance, when $q_0=\delta_{x_0}$ we have $f_u^{\text{int}}(t;\mu,\sigma, \mathcal{B},q_0)=f_u^{\text{single}}(t;\mu,\sigma,\mathcal{B},x_0)$ and similarly for $f_\ell^{\text{int}}$ and $q^{\text{int}}$, so that a single-stage model can be viewed as a special case of the integrated single-stage model. We use the same definitions and formulas in equation \eqref{eqn: int} when $q_0$ is a sub-probability density function, and this usage forms the basis of our multi-stage model below.

\subsection{Multi-stage Model}\label{sec: multi}
We define \textbf{multi-stage models} to be GDDMs with piecewise constant drift coefficient, piecewise constant diffusion coefficient, and continuous piecewise linear boundaries, which can be divided at their segment points such that each stage reduces to an integrated single-stage model. To be specific, a $d$-stage model is a stochastic process
\begin{equation*}
	\diff X(t)=\mu(t)\diff t + \sigma(t)\diff W(t)
\end{equation*}
stopped by a continuous piecewise linear upper boundary $u(t)$, a continuous piecewise linear lower boundary $\ell(t)$ and a vertical boundary $t=T_{\text{end}}$, where $\mu(t)$ is a piecewise constant drift function and $\sigma(t)$ is a piecewise constant diffusion function. Let $0<t_1<\cdots<t_{d-1}<T_{\text{end}}$ be the $d-1$ segment points for $\mu(t), \sigma(t), u(t)$ and $\ell(t)$ prior to $T_{\text{end}}$, and define $t_0= 0$ and $t_d = T_{\text{end}}$. Then $\mu(t), \sigma(t), u(t)$ and $\ell(t)$ can be defined as
\begin{equation*}
	\left.
	\begin{aligned}
		\mu(t)&=\mu_k \\
		\sigma(t)&=\sigma_k \\
		u(t)&=a_{k1}+b_{k1}(t-t_{k-1})\\
		\ell(t)&=a_{k2}+b_{k2}(t-t_{k-1})  
	\end{aligned}
	\right\}\text{when }t_{k-1} < t \le t_k
\end{equation*}
for $1\le k\le d$. Note that we must have $a_{k1}+b_{k1}(t_k-t_{k-1})=a_{k+1,1}$ and $a_{k2}+b_{k2}(t_k-t_{k-1})=a_{k+1,2}$ for $1\le k\le d-1$ to ensure the continuity of $u(t)$ and $\ell(t)$. We denote the process during time interval $\left(t_{k-1}, t_k\right]$ to be the $k$-th (integrated single) stage, which has a drift coefficient $\mu_k$, diffusion coefficient $\sigma_k$, and boundary parameters $\mathcal{B}_k=\left(a_{k 1}, b_{k 1}, a_{k 2}, b_{k 2}, T_k=t_k-t_{k-1}\right)$, for $1 \leq k \leq d$. By the Markov property, the NPD $q^{\text{int}}$ in equation \eqref{eqn: int} from the $(k-1)$-st stage becomes the initial position density $q_0$ of the $k$-th stage. Since the process could have hit the upper or lower boundary in an earlier stage, the NPD is a sub-probability density function, which is why we allowed $q_0$ to be a sub-probability density function in the definition of the integrated single-stage model. 

To compute the FPTDs and NPDs of a multi-stage model, we first divide the multi-stage model into several integrated single-stage models and then proceed sequentially: For the first stage, the initial position $X(0)$ has the same distribution as that of the considered multi-stage model (deterministic or random). Starting from the second stage, the NPD from the previous stage becomes the initial position distribution of the current stage. The FPTDs of each stage correspond exactly to the FPTDs of the original multi-stage model over the corresponding time interval. An illustration of the sequential computation is presented in \figref{fig: multi-stage-illustration}. In settings with the possibility of non-response before $T_{\text{end}}$, we can also compute the NPP $Q$ at $T_{\text{end}}$ by integrating the final NPD. Specifically, recalling that the density of $\tau$ on the upper and lower boundaries is denoted as $f_u(t)$ and $ f_{\ell}(t)$, respectively, then we can write the whole sequential algorithm as Algorithm \ref{alg: seq-multi}.  A formal derivation of the validity of the algorithm uses the Markov property and the Chapman-Kolmogorov equation with each single-stage model behaving like a sub-probability Markov transition kernel from one stage boundary to the next. We give a concrete example illustrating the complete computation process of applying Algorithm \ref{alg: seq-multi} to a multi-stage model in \figref{fig: aDDM_seq}.

\begin{figure*}[htb!]
	\centering
	\includegraphics[width=0.8\textwidth]{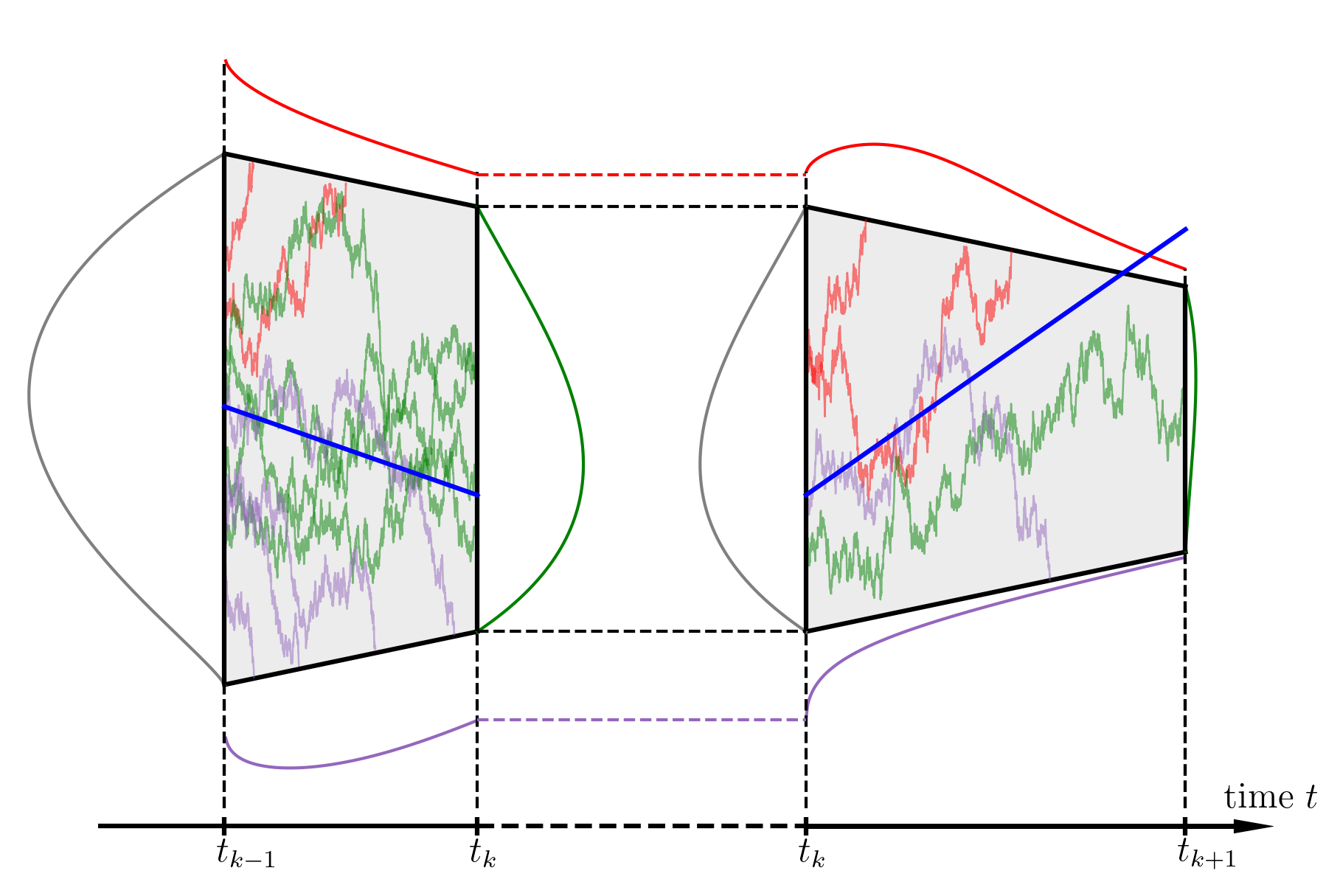}
	\caption{An illustration of the sequential computation from the $k$-th stage to the $(k+1)$-th stage in a multi-stage model. The $k$-th stage (depicted on the left side of the figure) is an integrated single-stage model, where the density of its starting position is given by $q_{k-1}(x)$ (represented by the grey curve). Within this stage, the functions $f_u(t)\mathbbm{1}_{(t_{k-1}, t_k]}(t)$ (red curve), $f_{\ell}(t)\mathbbm{1}_{(t_{k-1}, t_k]}(t)$ (purple curve), and $q_k(x)$ (green curve) are obtained by integrating the formulas \eqref{eqn: fptd_upper}, \eqref{eqn: fptd_lower}, and \eqref{eqn: fptd_vertical} over $q_{k-1}$, respectively. The resulting NPD $q_k(x)$ then serves as the starting position distribution for the $(k+1)$-th stage, facilitating the subsequent computation (depicted on the right side of the figure).
	}
	\label{fig: multi-stage-illustration}
\end{figure*}

\begin{algorithm*}[htb!]
	\SetKwInOut{Input}{Input}
	\SetKwInOut{Output}{Output}
	\caption{Sequential likelihood evaluation of a multi-stage GDDM}
	\label{alg: seq-multi}
	\Input{A multi-stage model configuration as described in Section \ref{sec: multi}. An observation $(t,c)$ for $0 < t \leq T_{\text{end}}$ and $c\in \{u, \ell\}$ or an observation of ``non-response''.}
	\Output{$f(t, c) = f_c(t)$ for a datum pair $(t,c)$, or $Q$ for ``non-response''}
	\BlankLine
	\For{$k = 1$ \KwTo $d$}{
		\begin{flalign*}
			&\left.\begin{aligned}
				f_u(t) &= f_u^{\text{int}}(t-t_{k-1};\mu_k, \sigma_k,\mathcal{B}_k, q_{k-1}) \\
				f_\ell(t) &= f_\ell^{\text{int}}(t-t_{k-1};\mu_k, \sigma_k,\mathcal{B}_k, q_{k-1}) \\
			\end{aligned}
			\right\}\text{ for } t\in(t_{k-1}, t_k]&\\
			&\, 
			q_k(x) = q^{\text{int}}(x;\mu_k, \sigma_k,\mathcal{B}_k, q_{k-1}) \text{ for all } x\in(\ell(t_k), u(t_k))&
		\end{flalign*}
	}
	\BlankLine
	$Q=\int_{\ell(t_d)}^{u(t_d)}q_d(x)\diff x$
\end{algorithm*}

For multi-stage models, Algorithm \ref{alg: seq-multi} is exact in the sense that the only error is the numerical error. We note, however, that equations \eqref{eqn: fptd_upper-basic} and \eqref{eqn: fptd_lower-basic} are singular at $t=0$, which can lead to numerical overflow/underflow in Algorithm \ref{alg: seq-multi} if the observation time $t\in(t_{k-1},t_k]$ is very close to $t_{k-1}$. In practical scenarios, this issue can usually be avoided via data post-processing that slightly inflates the value of any $t$ that happens too close to the beginning of a stage. This will have negligible impact on parameter estimation, particularly in cognitive scientific applications where any numerical corrections are going to lie well within the measurement error of participant reaction times.

\begin{figure*}[htb!]
	\centering
	\includegraphics[width=\textwidth]{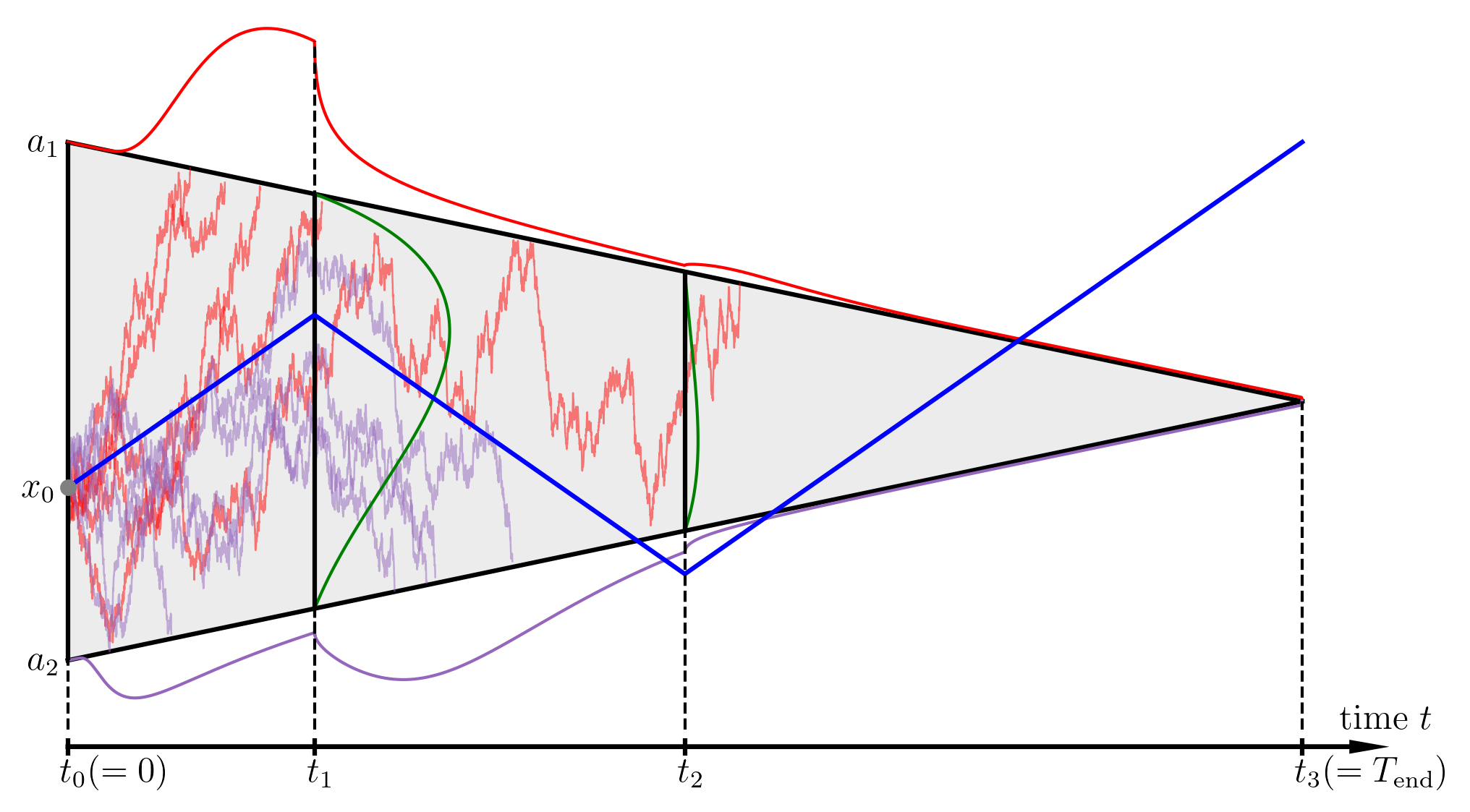}
	\caption{Example of applying the multi-stage algorithm to a GDDM with drift term $\mu_1\mathbbm{1}_{[0,t_1)\cup[t_2, \infty)}(t)+\mu_2\mathbbm{1}_{[t_1, t_2)}(t)$, constant $\sigma$ and linearly collapsing boundaries that collapse at $t_3$. The process is divided into 3 single-stage models according to $t\in[0, t_1),[t_1, t_2), [t_2, t_3)$. The FPTDs can be computed sequentially from the early to the late stages, and are plotted along their respective boundaries. Example trajectories are shown with different colors representing which boundary the trajectory first hits (red indicates upper boundary, purple indicates lower).}
	\label{fig: aDDM_seq}
\end{figure*}

\subsection{Approximating with the Multi-stage Model} \label{sec: approx}

For GDDMs of the same form as the multi-stage model, but with boundaries that are not piecewise linear, it is natural to consider approximating the GDDM with a multi-stage model by approximating the boundaries $u(t)$ and $\ell(t)$ with continuous piecewise linear functions. Choosing time points $0=t_0 < t_1 < \dotsb < t_d = T_{\text{end}}$ that include any times when $\mu(t)$ or $\sigma(t)$ changes values, we can define $\overline u(t)$ and $\overline \ell(t)$ to be the continuous piecewise linear functions that are linear on each interval $(t_{k-1},t_k]$ and for which $\overline u(t_k)=u(t_k)$ and $\overline \ell(t_k) = \ell(t_k)$ for each $k$. The original process with these new boundaries gives a multi-stage model that approximates the original GDDM. The FPTDs and NPDs of the approximating multi-stage model can be computed with Algorithm \ref{alg: seq-multi} and used as approximations of the FPTDs and NPDs for the original GDDM. 

Define $\overline\tau=\inf\{t > 0: X(t)\not\in(\overline\ell(t),\overline u(t))\}$ to be the hitting time of the process to the approximating boundaries. As the approximation gets increasingly fine, meaning $\max_k |t_k-t_{k-1}| \to 0$, then it is straightforward to show that $\overline\tau\to\tau$ with probability one under mild regularity conditions on the boundaries (see the proof in Appendix \ref{sec: approx_proof}). Showing that the FPTDs also converge is beyond the scope of this paper and presumably requires additional regularity assumptions on the boundaries, such as continuous differentiability. We demonstrate in Sections \ref{sec: example2} and \ref{sec: example3} that this approximation method can perform remarkably well in practice.

\subsection{Transformations}\label{sec: generalization}

For GDDMs that cannot be directly approximated with a multi-stage model as in Section \ref{sec: approx}, a transformation into an equivalent GDDM that is amenable to approximation may be possible.

Consider a transform $(t, x)\mapsto (s, w)$ defined as
\begin{equation}\label{eqn: transform}
	\begin{aligned} s & =\gamma(t),\\
    w & =\psi(x, t)\end{aligned}
\end{equation}
for $t\geq 0$, $x\in\mathbb{R}$, where $\gamma:[0,\infty)\to[0,D)$ is one-to-one, increasing, and differentiable with derivative $\gamma'$ for some $0 < D \leq \infty$, $\psi:\mathbb{R}{\times}[0,\infty)\to\mathbb{R}$ is continuous, and for each $t\geq 0$, $\psi(\cdot,t):\mathbb{R}\to\mathbb{R}$ is one-to-one, increasing, and differentiable with derivative $\psi'(\cdot,t)$ and inverse $\psi^{-1}(\cdot,t)$. The GDDM in the new coordinate system is
\begin{align*}
        \widetilde{X}(s)&=\psi(X(\gamma^{-1}(s)), \gamma^{-1}(s)) , \\
		\widetilde{u}(s)&=\psi(u(\gamma^{-1}(s)), \gamma^{-1}(s)) , \\
		\widetilde{\ell}(s)&=\psi(\ell(\gamma^{-1}(s)), \gamma^{-1}(s)) , \\
		\widetilde{\tau}&=\inf\{s>0: \widetilde{X}(s)\notin (\widetilde{\ell}(s), \widetilde{u}(s))\} , \\
		\widetilde{C}&= \text{exiting boundary} \in\{u,\ell\} 
\end{align*}
with the distribution of $\widetilde{X}(0)=\psi(X(0),0)$ defined by the distribution of $X(0)$. If $X(0)$ has density $q_0$, then $\widetilde{X}(0)$ has density $\widetilde{q}_0(x)=q_0(\psi^{-1}(x,0))/\psi'(\psi^{-1}(x,0))$.

The assumptions on $\gamma$ and $\psi$ ensure that $\widetilde{u}$ and $\widetilde{\ell}$ are continuous and that
\begin{equation*}
	\begin{aligned}
		\widetilde{\tau}&=\gamma(\tau) \quad \quad \text{and} \quad \quad \widetilde{C}=C .
	\end{aligned}
\end{equation*}
Hence, the FPTDs $f_c$ of $X$ are related to the FPTDs $\widetilde{f}_c$ of $\widetilde{X}$ via
\begin{equation}\label{eqn: relation}
\begin{aligned}
f_c(t)=&\frac{\diff}{\diff t}\mathbb{P}(\tau\le t,C=c)\\
=&\frac{\diff}{\diff t}\mathbb{P}(\widetilde{\tau}\le \gamma(t),\widetilde{C}=c)\\
=&\gamma'(t) \widetilde{f}_c(\gamma(t)) 
\end{aligned}
\end{equation}
for $c\in\{u,\ell\}$, 
and the NPD $q$ of $X$ at $T$ is related to the NPD $\widetilde{q}$ of $\widetilde{X}$ at $\gamma(T)$ via
\begin{equation} \label{eqn: relation-NPD}
	\begin{aligned}
		q(x) & = \frac{\diff}{\diff x}\mathbb{P}(X(T)\leq x,\tau > T)\\ &= \frac{\diff}{\diff x}\mathbb{P}(\widetilde{X}(\gamma(T))\leq \psi(x,T),\widetilde{\tau} > \gamma(T)) \\
		& = \psi'(x,T)\widetilde{q}(\psi(x,T)) .
	\end{aligned}
\end{equation}
The NPP $Q$ of $X$ at $T$ is simply the NPP $\widetilde{Q}$ of $\widetilde{X}$ at $\gamma(T)$.
Consequently, any method for approximating the FPTDs and NPDs of the transformed GDDM $\widetilde{X}$ can be used to approximate the FPTDs and NPDs of the original GDDM $X$ using equations \eqref{eqn: relation} and \eqref{eqn: relation-NPD}. The trick is to find a transformation for which the resulting $\widetilde{X}$ gives a GDDM that is amenable to approximation.

A particularly simple $\widetilde{X}$ is standard Brownian motion, perhaps with a random starting point, i.e., equation \eqref{eqn: sde} with $\mu(x,t)\equiv 0$ and $\sigma(x,t)\equiv 1$. A wide range of commonly used SDEs can be transformed into standard Brownian motion in this way, including geometric Brownian motion, the Ornstein-Uhlenbeck (O-U) process, and the Brownian bridge, among others. The question of whether a diffusion process can be transformed to a Brownian motion through transformation of the type in equation \eqref{eqn: transform} was raised by Kolmogorov in \cite{kolmogoroff1931analytischen}. Cherkasov gave a constructive proof of the necessary and sufficient condition for the existence of such a transformation in \cite{cherkasov1957transformation}, which is now commonly referred to as the Cherkasov condition in the literature. This was recast into a simpler form in \cite{ricciardi1976transformation}, and was later used in \cite{ricciardi1983note} to obtain the one-sided first passage time density through the integral equation method. We present the Cherkasov condition from \cite{ricciardi1976transformation} in Appendix \ref{appendix: cherkasov}, and instead of focusing on the condition itself, we provide specific examples of the transformation from diffusion processes to Brownian motion that are common in practical applications.

    \textbf{DDMs with nonlinear drift.}
	One simple but important subclass of models that fits this transformation is
	\begin{equation}\label{eqn: general_ddm}
		\diff X(t)=\mu(t) \diff t +\sigma \diff W(t)
	\end{equation}
	with possibly nonlinear time-dependent drift term $\mu(t)$. The SDE \eqref{eqn: general_ddm} can be solved via a direct integration:
	\begin{equation*}
		X(t)=X(0)+\int_0^t \mu(r) \diff r +\sigma W(t)
	\end{equation*}
	and consequently, through transformation $w=\psi(x,t)=\frac{1}{\sigma}(x-\int_0^t \mu(r)\diff r)$ (in this case $s=\gamma(t)=t$, so we still use $t$) we have
	\begin{equation}\label{eqn: transformed_ddm}
		\widetilde{X}(t) = \psi(X(t), t)=\frac{X(0)}{\sigma} + W(t) .
	\end{equation}
	The right hand side of \eqref{eqn: transformed_ddm} is a Brownian motion with initial position $X(0)/\sigma$. We illustrate the approximate FPTD computation for this example in Section \ref{sec: example2}.
	
    \textbf{Ornstein-Uhlenbeck model.} The nonstationary Ornstein-Uhlenbeck (O-U) process $X(t)$ is defined by the following stochastic differential equation:
	\begin{equation}\label{eqn: ou}
		\diff X(t)=\theta\left(\lambda-X(t)\right) \diff t+\sigma \diff W(t)
	\end{equation}
	where $\theta\neq 0, \lambda\in\mathbb{R}, \sigma>0$ are constants. When $\theta > 0$ it models the ``mean-reverting" phenomenon as the process tends towards $\lambda$ over time, and it will converge to its stationary distribution $\mathcal{N}(\lambda, \frac{\sigma^2}{2 \theta})$ as $t$ goes to infinity. This process is used, for instance, in physics to model the velocity of a particle under friction, in finance to capture mean-reverting behavior of interest rates and volatility \citep{gardiner2009stochastic}, and in cognitive science to model evidence accumulation and neural dynamics \citep{usher2001time,Brunton2013}.

	The solution of \eqref{eqn: ou} is given by \cite{gardiner2009stochastic}
	\begin{equation*}
		X(t)=\lambda+e^{-\theta t}(X(0)-\lambda)+\sigma e^{-\theta t}\int_0^t e^{\theta r} \diff W(r) . 
	\end{equation*}
	The stochastic integral
	can be viewed as a time-changed Brownian motion \citep{doob1942brownian}:
	\begin{equation}\label{eqn: view_time_changed}
		X(t)=\lambda+e^{-\theta t}(X(0)-\lambda)+\sigma e^{-\theta t}\widetilde{W}\Big(\frac{e^{2\theta t}-1}{2\theta}\Big) ,
	\end{equation}
	where 
    \begin{equation*}
    \widetilde{W}(s)=\int_0^{\frac{1}{2\theta} \log (1+2\theta s) }e^{\theta r}\diff W(r)
    \end{equation*}
    is another standard Brownian motion defined by the above stochastic integral.
	We perform the transformation $s=\gamma(t)=(e^{2\theta t}-1)/(2\theta)$ and $w=\psi(x, t)=(e^{\theta t} x-\lambda e^{\theta t}+\lambda)/\sigma$
	so that the transformed process becomes
	\begin{equation*}
		\widetilde{X}(s) = \psi(X(\gamma^{-1}(s)), \gamma^{-1}(s))=\frac{X(0)}{\sigma}+\widetilde{W}(s)
	\end{equation*}
	which is a Brownian motion with initial position $X(0)/\sigma$.

	In Section \ref{sec: example3}, we use this transformation and adopt the procedures in this section to compute the FPTDs of the O-U process.
Additionally, for parametric models like the O-U process, we can also study their ``multi-stage" versions with piecewise constant parameters. Our Algorithm \ref{alg: seq-multi} naturally admits generalization to these models, where the stage breakpoints comprise both the original parameter breakpoints and the segmentation points introduced to approximate the nonlinear transformed boundaries.

Even if the original model has simple boundaries, the transformed model could have complicated boundaries that make approximations based on the multi-stage model less accurate or risk numerical overflow/underflow if an interpolation knot point $t_k$ is placed too close to a reaction time of interest. For instance, the time transformation $s=\gamma(t)$ in the O-U process is exponential and might lead to time dilation issues if one is interested in the FPTDs at large decision times.

\subsection{Algorithm Summary and Implementation} \label{sec: summary}

The complete process of approximating $f(t,c)$ and/or $Q$ is as follows:
\begin{itemize}
	\item If necessary, use the methods in Section \ref{sec: generalization} to mathematically transform the GDDM into a representation with constant or piecewise constant $\mu(t)$ and $\sigma(t)$.
	\item If necessary, use the strategy in Section \ref{sec: approx} to approximate the upper and lower boundaries with piecewise linear functions, so that the approximate model is a multi-stage model.
	\item Use Algorithm \ref{alg: seq-multi} in Section \ref{sec: multi} to calculate the FPTD or the NPP of the multi-stage model.
	\item If necessary, use the methods of Section \ref{sec: generalization} to transform the resulting FPTD from the multi-stage model back into the FPTD of the original GDDM.
\end{itemize}

In our implementation, all numerical integrations in equation \eqref{eqn: int} as used by Algorithm \ref{alg: seq-multi} are performed with Gauss-Legendre quadrature (see Chapter 5.5 of \cite{sauer2011numerical} for an introduction). This approach offers higher accuracy with fewer integrand evaluations compared to alternative methods, such as the trapezoidal rule. We specifically benefit from this in our method because our integrands consist of (truncated) infinite series, which are expensive to evaluate. It also means that each $q_k$ within the algorithm needs only to be represented at a small number of predetermined quadrature points. The choice of quadrature order provides a natural trade-off between speed and accuracy: higher orders increase the number of quadrature points and hence the computational cost, but also yield greater precision. Details of the implementation of Algorithm \ref{alg: seq-multi} 
	regarding quadrature are in Appendix \ref{appendix: quad}.

There are several other important practical aspects of Algorithm \ref{alg: seq-multi} that contribute to its efficiency, which we discuss in detail in Appendix \ref{sec: implementation}. In our implementation, the core functionality is written in Python, while performance-critical components are optimized using Cython \citep{behnel2011cython}, leveraging both the flexibility of Python and the high efficiency of C. OpenMP parallelization is supported to compute likelihoods for different trials independently, thereby accelerating the computation on large datasets. Our code, along with usage examples, including the numerical examples presented in Section \ref{sec: examples}, will be released as an open-source Python package upon publication for reproducibility. Additionally, our code will be integrated into the HSSM package \citep{hssm2025} for hierarchical Bayesian inference of GDDMs.

\section{Numerical Examples}\label{sec: examples}

In this section, we provide several numerical examples to demonstrate the effectiveness of our method. In Examples \ref{sec: example1}, \ref{sec: example2}, and \ref{sec: example3}, we compute exact or approximate FPTDs using Algorithm \ref{alg: seq-multi} and verify the results by comparing them with simulated data. In Section \ref{sec: stat_addm}, we demonstrate the efficiency of Algorithm \ref{alg: seq-multi} for likelihood computation, highlighting its speed relative to previous methods and its applicability to statistical inference.

It is worth mentioning that accurately simulating first passage times of stochastic processes is itself a challenging problem (see \cite{tuerlinckx2001comparison} for a discussion regarding Brownian motion). The usual approach, which we have also adopted, is to discretize the SDE \eqref{eqn: general_ddm} according to the Euler-Maruyama scheme: let $0<r_1<r_2<\cdots$ be a temporal grid, simulate $\widehat{X}(0)$ according to the distribution of $X(0)$, then proceed as
\begin{equation}\label{eqn: euler-maruyama}
\begin{aligned}
&\widehat{X}(r_{n+1})=\widehat{X}(r_n)+\mu(\widehat{X}(r_n), r_n)(r_{n+1}-r_n)+ \sigma(\widehat{X}(r_n), r_n) \Delta W(r_n)
\end{aligned}
\end{equation}
to simulate an approximate sample path of $X$, where the random variable $\Delta W(r_n)= W(r_{n+1})-W(r_n)$ is distributed as $\mathcal{N}(0, r_{n+1}-r_n)$, and all $\Delta W(r_n)$'s are independent. Such a simulation will, of course, be biased due to the unavoidable discretization error. Then, the first passage time data is the first time $r_n$ such that $\widehat{X}(r_n)\notin (\ell(r_n), u(r_n))$. However, even if we could simulate $X$ exactly on the discrete grid $(r_n)_{n=1}^\infty$, this would not sample $\tau$ exactly, as it is possible for $X$ to cross the barrier between grid points $r_n$ and $r_{n+1}$ without exceeding the barrier at either of these grid points. 

In our experiments, we simply use a very small step size $\Delta r=r_{n+1}-r_n=10^{-5}$ and a large sample size to reduce the approximation error. More advanced techniques, such as Brownian bridge interpolation, are also commonly used to address this issue (see chapters 6-7 of \cite{glasserman2004monte} for a general introduction), however for the purposes of this work, the simulation data serves as a sanity-check more than an object of independent interest.

\subsection{FPTDs for GDDMs with Piecewise Constant Drift}\label{sec: example1}
We consider a GDDM with piecewise constant drift
\begin{equation*}
	\mu(t)=\begin{cases}
		1 &0< t\le 1\\-0.2 &1< t\le 2.5 \\1.5 &2.5< t\le 3.5\\0.5 &3.5< t\le 4\\-1 & 4< t\le 5
	\end{cases}
\end{equation*}
constant diffusion $\sigma=1$ and linear collapsing boundaries:
\begin{equation*}
	u(t)=1.5-0.3t,\quad \ell(t)=-1.5+0.3t
\end{equation*}
starting at $x_0=-0.5$. This is a generic example of a multi-stage model with the same drift terms for all trials. 

We use our Algorithm \ref{alg: seq-multi} to compute the FPTDs on the upper boundary and lower boundary and compare them with the normalized histogram obtained from 50,000 Monte Carlo samples. 
The results are shown in \figref{fig: addm_llhd}. The FPTD exhibits corners at all breakpoints of $\mu(t)$. Despite this complexity, our result visually shows excellent agreement with the empirical densities.
\begin{figure*}[htb!]
	\centering
	\includegraphics[width=0.8\textwidth]{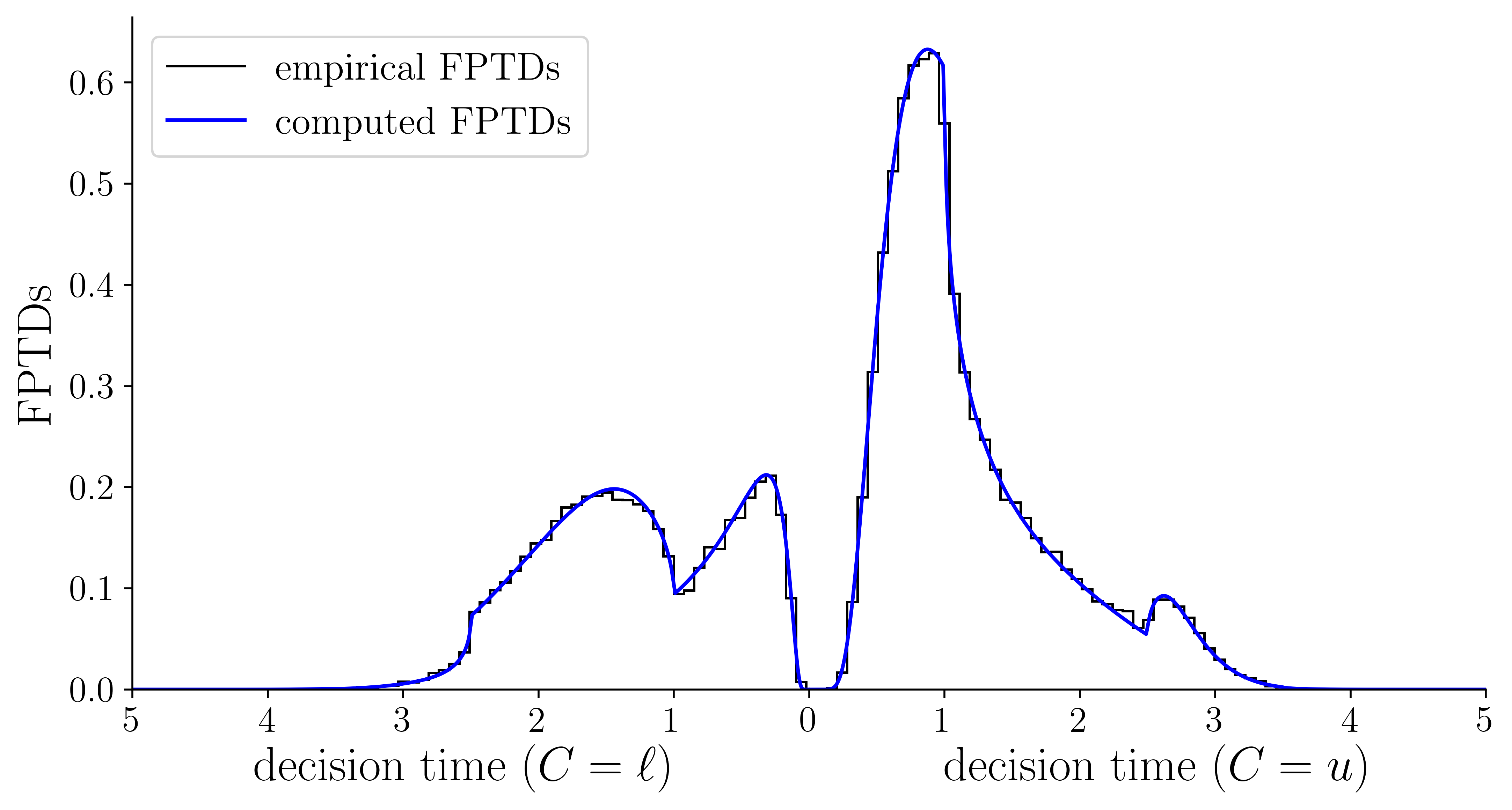}
	\caption{Plots of FPTDs and histogram for Example \ref{sec: example1}: a GDDM with piecewise constant drift. First passage times on the upper boundary are displayed on the positive $t$-axis and those on the lower boundary are displayed on the negative $t$-axis so that the whole plot is a valid probability density function.}
	\label{fig: addm_llhd}
\end{figure*}

\subsection{FPTDs for GDDMs with Nonlinear Drift and Boundaries}\label{sec: example2}
Consider the GDDMs in equation \eqref{eqn: general_ddm} with possibly nonlinear time-dependent drift term $\mu(t)$ and boundaries $u(t), \ell(t)$. Following the transform discussed in Section \ref{sec: generalization}, we can write the transformed boundaries as
\begin{equation*}
	\begin{aligned}
		\widetilde{u}(t)&=\frac{1}{\sigma}\Big(u(t)-\int_0^t \mu(r) \diff r\Big),\\
		\widetilde{\ell}(t)&=\frac{1}{\sigma}\Big(\ell(t)-\int_0^t \mu(r) \diff r\Big)
	\end{aligned}
\end{equation*}
and work with the transformed model
$\widetilde{X}(t)=\frac{X(0)}{\sigma} + W(t)$.
We approximate $\widetilde{u}(t)$ and $\widetilde{\ell}(t)$ by piecewise linear functions and use Algorithm \ref{alg: seq-multi} to evaluate the FPTD of the approximate multi-stage model. This can apply to any nonlinear drift and boundaries.

We show this case via a GDDM with a sinusoidal drift
\begin{equation*}
	\mu(t)=\frac{1}{2}\sin t,
\end{equation*}
constant diffusion $\sigma=1$, and boundaries
\begin{equation*}
	u(t)=2e^{-\left(t/5\right)^3},\quad \ell(t)=-2e^{-\left(t/5\right)^3}
\end{equation*}
starting at $x_0=-0.5$. The upper boundary $u(t)$ is the survival function of a Weibull$(5,3)$ distribution, and $\ell(t)$ is its reflection about the $t$-axis. The choice of a Weibull distribution is guided by relevant examples in prior work \citep{fengler2021likelihood,GingJehli2025}.

The original boundaries, and the transformed boundaries with their piecewise linear approximations are shown in \figref{fig: weibull_llhd}(a) and \figref{fig: weibull_llhd}(b), respectively. The computed approximate FPTDs and the histogram consisting of 50,000 samples are displayed in \figref{fig: weibull_llhd}(c). The approximate FPTD and the histogram agree well.

\begin{figure*}[htb!]
	\centering
	\includegraphics[width=\textwidth]{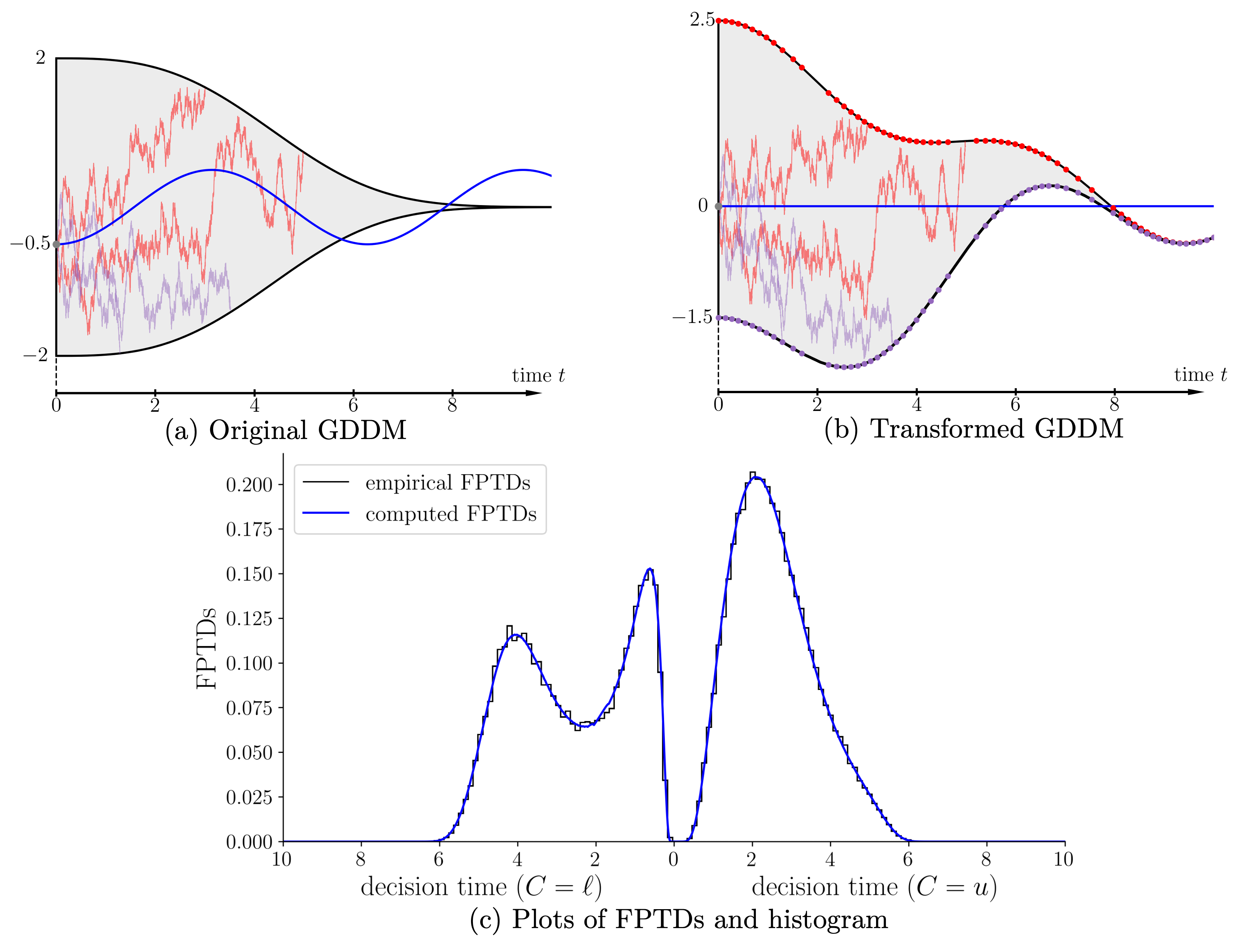}
\caption{Example \ref{sec: example2}: a GDDM with sinusoidal drift and Weibull survival boundaries. (a) and (b) illustrate the GDDM before and after transformation. The expected value of the process (blue curves) is transformed from a sinusoidal function to a constant zero. In (b), red dots on the upper boundary and purple dots on the lower boundary mark the knot points used to form piecewise linear boundary approximations. These knots are selected via adaptive linear interpolation (see Appendix \ref{sec: implementation} for a detailed description). (c) displays first passage times: those on the upper boundary appear on the positive $t$-axis, while those on the lower boundary appear on the negative $t$-axis, forming a valid probability density function.}
	\label{fig: weibull_llhd}
\end{figure*}

To quantitatively justify the approximation, we conduct a one-sample Kolmogorov-Smirnov (KS) test. The null hypothesis $H_0$ is that the sampled first passage time data are distributed according to the computed approximate FPTDs. As shown in Table \ref{tab: ks-test1}, the KS test does not suggest rejecting $H_0$ even for sample sizes as large as 100,000. Even though our method of simulating first passage times introduces discretization errors, we believe that these results demonstrate the high accuracy of the FPTD approximation.

\begin{table}[htb!]
	\centering
	\label{tab:ks-test-combined}
    \caption{Results of the one-sample Kolmogorov-Smirnov tests for different data sizes.}
		\centering
		\begin{tabular*}{\columnwidth}{@{\hspace{\tabcolsep}\extracolsep{\fill}}ccccc@{\hspace{\tabcolsep}}}
			\toprule
			\textbf{\#Data} &$10000$& $20000$& $50000$&$100000$\\
			\midrule
			\textbf{$p$-value} & $0.778$&$0.703$& $0.498$& $0.499$ \\
			\bottomrule
		\end{tabular*}
		\subcaption{Example \ref{sec: example2}: a GDDM with sinusoidal drift and Weibull survival boundaries}
		\label{tab: ks-test1}
	
	\vspace{1em} 
	
		\centering
		\begin{tabular*}{\columnwidth}{@{\hspace{\tabcolsep}\extracolsep{\fill}}ccccc@{\hspace{\tabcolsep}}}
			\toprule
			\textbf{\#Data} &$10000$& $20000$& $50000$&$100000$\\
			\midrule
			\textbf{$p$-value} & $0.449$&$0.473$& $0.179$& $0.357$ \\
			\bottomrule
		\end{tabular*}
		\subcaption{Example \ref{sec: example3}: an Ornstein-Uhlenbeck model}
		\label{tab: ks-test2}
\end{table}

\subsection{FPTDs for Ornstein-Uhlenbeck Models}\label{sec: example3}
Consider the Ornstein-Uhlenbeck model in equation \eqref{eqn: ou} with boundaries $u(t)$ and $\ell(t)$. According to the transformation given in Section \ref{sec: generalization}, the transformed boundaries are given by
	\begin{align*}
		\widetilde{u}(s)=&\tfrac{1}{\sigma }\Big(u\big(\tfrac{1}{2\theta}\log(1+2\theta s)\big)\sqrt{1+2\theta s}-\lambda\sqrt{1+2\theta s}+\lambda\Big),\\
		\widetilde{\ell}(s)=&\tfrac{1}{\sigma }\Big(\ell\big(\tfrac{1}{2\theta}\log(1+2\theta s)\big)\sqrt{1+2\theta s}-\lambda\sqrt{1+2\theta s}+\lambda\Big)
	\end{align*}
We illustrate this case with an example that uses random initial condition $X(0)\sim\text{Beta}(10, 25)$. We set the model parameters as $\theta=1,\lambda=1.5,\sigma=2$, and the original boundaries to be
\begin{equation*}
	u(t)=2e^{-\left(t/2\right)^3},\quad \ell(t)=-2e^{-\left(t/2\right)^3}
\end{equation*}
We follow the procedures described in Section \ref{sec: summary}: interpolate the transformed boundaries sufficiently fine, compute the approximate transformed FPTDs by Algorithm \ref{alg: seq-multi}, and obtain the original FPTDs by relation \eqref{eqn: relation}. Plots of the FPTDs and statistical test results are shown in \figref{fig: ou_llhd} and Table \ref{tab: ks-test2}, respectively. From the results we conclude that our method effectively computes the FPTDs of the O-U model \eqref{eqn: ou}.

\begin{figure*}[htb!]
	\centering
	\includegraphics[width=0.8\textwidth]{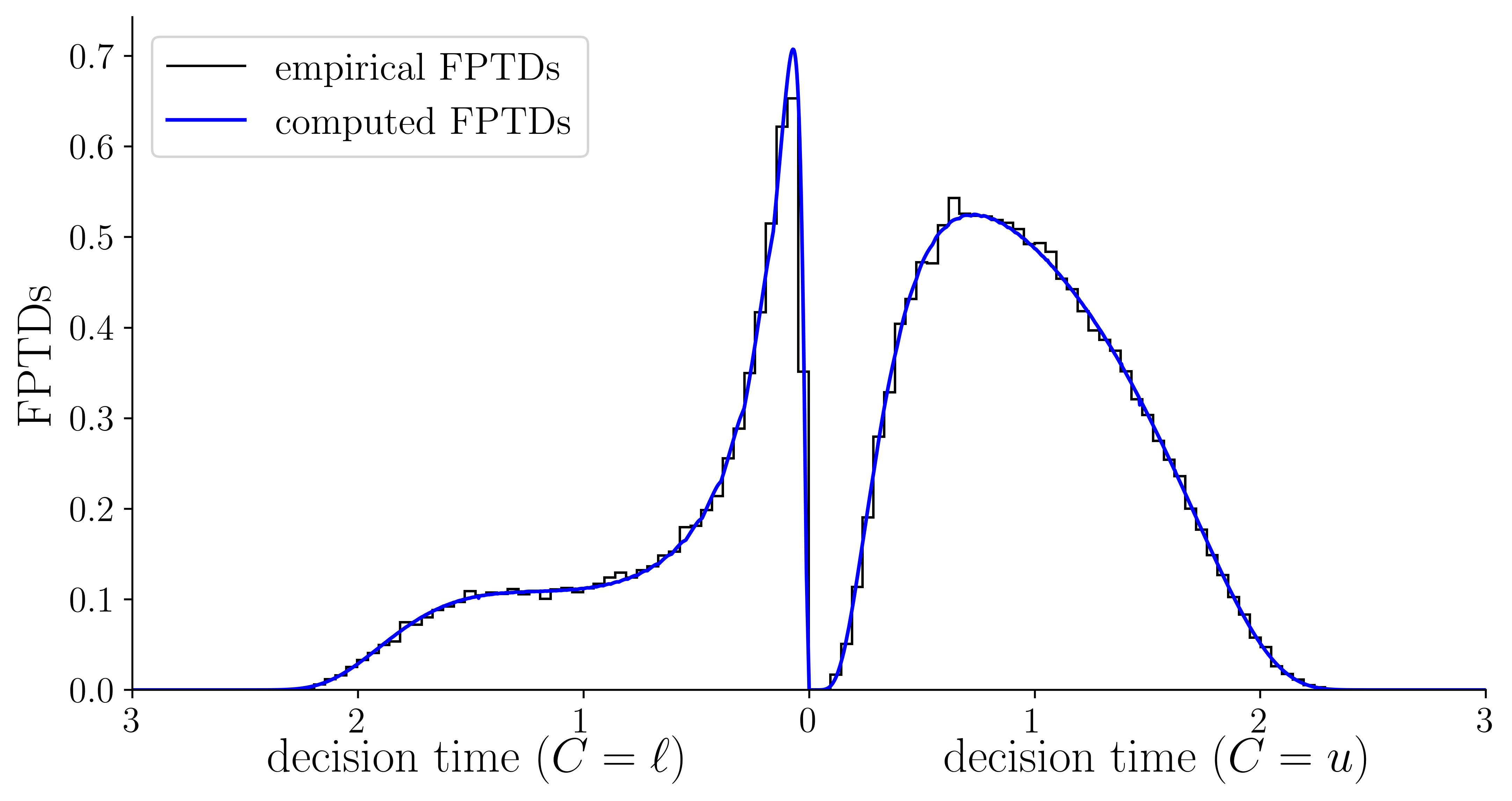}
	\caption{Plots of FPTDs and histogram for Example \ref{sec: example3}: an Ornstein-Uhlenbeck model. First passage times on the upper boundary are displayed on the positive $t$-axis and those on the lower boundary are displayed on the negative $t$-axis so that the whole plot is a valid probability density function.}
	\label{fig: ou_llhd}
\end{figure*}

\subsection{Practical Statistical Computing for aDDMs}\label{sec: stat_addm}
In the following examples, we show the effectiveness and efficiency of our method for computing the likelihood and inferring parameters for aDDMs. The aDDM incorporates visual attention as a key factor underlying the decision process behind simple choice behavior. It adjusts the drift rate $\mu(t)$ based on instantaneous attention within a given experimental trial, adopting different values depending on whether the participant fixates on choice options A or B at time $t$. Let $V(t)$ be the process taking values in $\{\mathrm{A}, \mathrm{B}\}$ representing the participant's visual fixation at time $t$, and let $r_A$ and $r_B$ denote the participant's numerical ratings of their preferences for options A and B, respectively. The aDDM corresponds to the following SDE (conditioned on $V, r_A, r_B$) \citep{krajbich2010visual},
\begin{equation}\label{eqn: addm}
	\begin{aligned}
		&\diff X(t)=A(V(t), r_A, r_B) \diff t+\sigma \diff W(t)
	\end{aligned}
\end{equation}
where
\begin{equation*}
A(v, r_A, r_B)=\begin{cases}
			\kappa(r_A-\eta r_B) & \text{if }v=\mathrm{A} ,\\
			\kappa(\eta r_A-r_B) & \text{if }v=\mathrm{B} ,\\
		\end{cases}
\end{equation*}
$\eta\in(0, 1)$ is a parameter that reflects the discounting of the non-fixated option, and $\kappa>0$ parametrizes the overall magnitude of the drift. Note that $\mu(t)=A(V(t), r_A, r_B)$ is a piecewise constant function as in the multi-stage model of Section \ref{sec: multi}, but that it will generally be different on each experimental trial, both because a participant's visual fixation process will generally be different on each trial and because the options presented to the participant can change across trials.

In an experimental setting, a participant's fixation trajectory $V(t)$ is captured using an eye-tracker, while the participant's preference ratings $r_A$ and $r_B$ are collected prior to the choice experiment; together, these can be viewed as known covariates. Because those covariates drive time-varying drift functions, the configuration $\boldsymbol{\Xi}$ of the GDDM will be different on each experimental trial. Letting $\boldsymbol{\Xi}_i$ denote the configuration on the $i$-th trial, then the likelihood of $\boldsymbol{\theta}$ is 
\begin{equation} \label{eqn: likeADDM}
	\mathcal{L}_{\text{aDDM}}(\boldsymbol{\theta}\mid((\tau_i, C_i))_{i=1}^n)=\prod_{i=1}^n f(\tau_i, C_i\mid\boldsymbol{\Xi}_i) .
\end{equation}
$\boldsymbol{\Xi}_i$ depends on unknown parameters $\boldsymbol{\theta}$ (e.g. $\eta$ and $\kappa$) that are shared by all trials, as well as on the trial-specific known visual fixation trajectory $V_i(t)$ and ratings $r_{A, i}, r_{B, i}$ for the subject's $i$-th trial.

Here, we simulate a collection of $n=50,000$ visual fixation trajectories $(V_i(t))_{i=1}^n$, preference ratings $((r_{A,i}, r_{B, i}))_{i=1}^n$ and the corresponding dataset of choices and reaction times $((\tau_i, C_i))_{i=1}^n$ from a prescribed aDDM model with drift parameters $\kappa, \eta$, constant diffusion $\sigma=1$ and symmetric linear collapsing boundaries 
\begin{equation*}
	u(t)=a-bt,\quad \ell(t)=-a+bt
\end{equation*} starting at $x_0 \in (-a,a)$. The fixation trajectory $V_i(t)$ on each trial was sampled from a gamma process to mimic the switching of attention in psychological experiments, and the ratings $r_{A, i}, r_{B, i}$ are sampled from a discrete uniform distribution on $\{1,2,3,4,5\}$. The average number of fixations per trial in our simulation is 3.45. 

\subsubsection{Likelihoods for aDDM given a Dataset}\label{sec: timing}
We evaluate the performance of Algorithm \ref{alg: seq-multi} for computing likelihoods of the parameters at their true values given the dataset $((\tau_i, C_i))_{i=1}^n$. We compare the speed of our methods with PyDDM \citep{shinn2020flexible}, a prominent and highly optimized library which we use as a benchmark on speed for the KFE approach (Section \ref{sec: related_work}). When using PyDDM, we can choose the finite difference discretization parameters $\Delta t$ and $\Delta x$ which determine the fineness of the temporal and spatial meshes, and as $\Delta t\rightarrow 0, \Delta x\rightarrow 0$, the numerical solution converges to the true solution of the underlying KFE governing the diffusion process. We test Algorithm \ref{alg: seq-multi} with varying quadrature orders and compare its speed–accuracy trade-off against the KFE method evaluated with different choices of $\Delta t$ and $\Delta x$. The results are shown in Figure \ref{fig: speed_accu}.

\begin{figure*}[htb!]
	\centering
	\includegraphics[width=\textwidth]{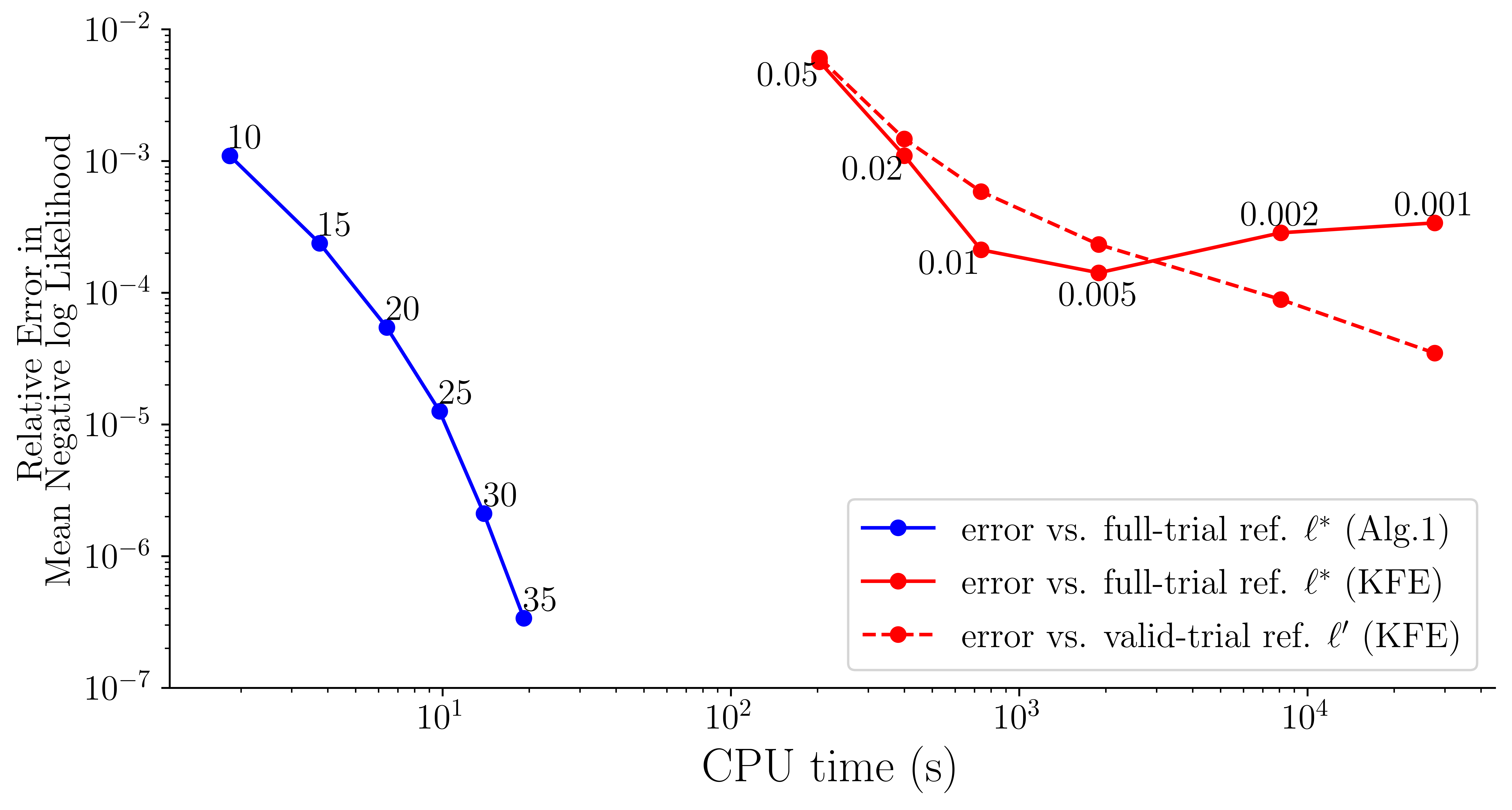}
	\caption{Speed–accuracy trade-off for Algorithm~\ref{alg: seq-multi} (blue) and the KFE method as implemented in PyDDM (red solid and dashed). Algorithm~\ref{alg: seq-multi} performance is shown for increasing quadrature order from $10$ to $35$. KFE performance is shown for decreasing space-time discretization ($\Delta t=\Delta x$) from $0.05$ to $0.001$. In all cases, we used the corresponding method to compute the FPTDs for all trials: $(f(\tau_i, C_i\mid\boldsymbol{\Xi}_i))_{i=1}^n$, where $n=50{,}000$ and the parameters in $\boldsymbol{\Xi}_i$ are those that generated the data. CPU time in each case is the total time to compute FPTDs for all 50,000 trials. The FPTDs can be combined to give the trial-normalized negative log-likelihood $\ell^*=-\frac{1}{n}\sum_{i=1}^n \log f(\tau_i, C_i\mid\boldsymbol{\Xi}_i)$ (see equation \eqref{eqn: likeADDM}). Since the aDDM in this example corresponds exactly to a multi-stage GDDM, we used Algorithm~\ref{alg: seq-multi} with a high quadrature order ($200$ quadrature points per stage) to obtain reference values for $\ell^*$ and $\ell'$ (defined below). Relative errors for Algorithm~\ref{alg: seq-multi} (blue line, vertical axis) are reported against this reference value of $\ell^*$.
    In four trials the KFE implementation (with these $\Delta t,\Delta x$) gave FPTDs of zero, yielding an infinite approximation of $\ell^*$. To provide a more meaningful sense of error for KFE, we also calculated the likelihood excluding those trials, namely $\ell'=-\frac{1}{n-4}\sum_{\text{valid $i$}} \log f(\tau_i, C_i\mid\boldsymbol{\Xi}_i)$, where the ``valid'' $i$ are the $49{,}996$ trials for which KFE gave strictly positive FPTD approximations. The solid red line compares the KFE approximation of $\ell'$ to the reference value of $\ell^*$. The error plateaus because KFE is forced to exclude the trials where its approximation fails. The dashed red line compares the KFE approximation of $\ell'$ to the reference value of $\ell'$. Although this is the wrong reference value for the target problem, it illustrates that the KFE approximation continues to improve as the discretization gets finer. 
}
	\label{fig: speed_accu}
\end{figure*}

    

Figure \ref{fig: speed_accu} demonstrates that our method outperforms the KFE method by more than \textit{three orders of magnitude} at a given target accuracy level. Moreover, our method is more numerically stable when evaluating trials with small likelihoods. PyDDM returned a likelihood of zero for several trials, complicating its use for likelihood-based methods.

The speedup of our method for aDDMs largely stems from its ability to progress in large time steps aligned with saccade times, leading to extremely fast single trial computation speed. The accurate analytical formulas in equations \eqref{eqn: fptd_upper-basic}--\eqref{eqn: fptd_vertical-basic} and the high-order efficient quadrature rules in equation \eqref{eqn: GL-quad} further contribute to our method's speed. In contrast, the KFE method requires much finer time discretization in the backward Euler scheme and therefore substantially more floating-point operations to attain comparable accuracy. 

\subsubsection{Statistical Inference of aDDM Parameters}
Here we conduct parameter recovery studies for the parameters $\boldsymbol{\theta} = (\eta, \kappa, a, b,x_0)$ in both the frequentist and Bayesian settings. We use the maximum likelihood estimator (MLE) 
\begin{equation*}
\begin{aligned}
\widehat{\boldsymbol{\theta}}_{n}^{\text{\sc mle}}&=\argmax_{\boldsymbol{\theta}} \mathcal{L}_{\text{aDDM}}(\boldsymbol{\theta}\mid((\tau_i, C_i))_{i=1}^n)=\argmax_{\boldsymbol{\theta}} \prod_{i=1}^n f(\tau_i, C_i\mid\boldsymbol{\Xi}_i)
\end{aligned}
\end{equation*}
as the frequentist approach. In the Bayesian framework, we place a prior $f_{\text{prior}}(\boldsymbol{\theta})$ on the parameters and approximate the posterior density
\begin{equation*}
\begin{aligned}
f_{\text{posterior}}(\boldsymbol{\theta})\propto f_{\text{prior}}(\boldsymbol{\theta}) \prod_{i=1}^n f(\tau_i, C_i\mid \boldsymbol{\Xi}_i)
\end{aligned}
\end{equation*}
using Markov chain Monte Carlo (MCMC). For both approaches we use Algorithm \ref{alg: seq-multi} to evaluate the likelihood. The true and estimated parameters are shown in Table \ref{tab: inference}.

\begin{table}[htb!]
	\centering
    \caption{True values, point estimates and interval estimates of aDDM parameters. The maximum likelihood estimates are reported together with component-wise nonparametric pivotal bootstrap confidence intervals based on 1000 resamples. The optimization is carried out using the trust region method. For the Bayesian approach, we report the posterior means and 95\% Bayesian credible intervals. The MCMC sampling was performed using the Metropolis–Hastings algorithm with $10^4$ burn-in iterations followed by $10^5$ draws. The MLE, posterior means and medians (not reported) coincide up to three decimal places. Note that all true parameter values lie within their corresponding 95\% confidence intervals and credible intervals.}
    \label{tab: inference}
    \begin{tabular}{cccccc}
		\toprule
		& True & MLE & 95\% confidence interval & posterior mean & 95\% Bayesian credible interval\\
		\midrule
		$\eta$ & $0.7$ & $0.697$ & $[ 0.690,  0.705]$ & $0.697$ & $[0.690,  0.704]$ \\
		$\kappa$ & $0.5$ & $0.502$ & $[ 0.497,  0.507]$ & $0.502$ & $[0.497,  0.507]$  \\
		$a$ & $2.1$ & $2.104$ & $[ 2.093,  2.116]$ & $2.103$ & $[2.092,  2.115]$\\
		$b$ & $0.3$ & $0.304$ & $[ 0.300,  0.309]$ & $0.304$ &$[0.300,  0.308]$ \\
		$x_0$ & $-0.2\ \ $ & $-0.196\ \ $ & $[-0.205, -0.188]$ & $-0.196\ \ $ &$[-0.205, -0.188]$ \\
		\bottomrule
	\end{tabular}
\end{table}

For maximum likelihood estimation, the numerical optimization takes around 1 minute when run with 32 threads, a computation time that is readily achievable even on a standard desktop computer. For Bayesian inference, an MCMC run with $10^5$ draws completes in $10$ hours. In contrast, because a single round of accurate likelihood evaluation with the KFE approach alone requires more than $3$ hours, the numerical optimization would take days or more likely weeks, and Bayesian methods would be entirely out of reach. This illustrates that our method enables fast and accurate parameter recovery which no previous methods are capable of.

\section{Discussion}\label{sec: conclusion}
In this paper, we propose a novel numerical method for computing the first passage time densities (FPTDs), non-passage densities (NPDs), and non-passage probabilities (NPPs) of a general class of drift diffusion models (DDMs) with two time-dependent boundaries. We utilize analytical formulas for certain single-stage models, which are then flexibly combined via numerical quadrature to accommodate a multi-stage framework. We then generalize to prominent, but previously computationally-intractable models in cognitive science \citep{krajbich_attentional_2012, krajbich2010visual} and beyond. Our method beats current alternatives by over three orders of magnitude on speed of execution in relevant application scenarios, while maintaining desired accuracy, as shown in Section \ref{sec: stat_addm}. 

The benefits over existing approaches are especially notable for the likelihood computation of GDDMs with trial-specific covariates such as the aDDM. Despite strong scientific interest, applications of aDDMs have been constrained by computational challenges. Current approaches to inference in aDDMs rely on unjustified computational simplifications, the so-called time-averaged drift approximations (TADA); see, e.g., \citep{lombardi2021piecewise, molter2019glambox}. The TADA replaces the time-varying drift rate $\mu(t)$ in a single trial of the aDDM with the constant drift rate $\tau^{-1}\int_0^\tau \mu(t)\diff t$ that is simply the temporal average of the aDDM drift rate on that trial. Although the TADA enables fast parameter inference using analytic formulas, it is not a well-specified statistical model, since the drift rate on a trial depends on the decision time on that trial. In parallel work \citep{liu2026time}, we show that maximum likelihood estimators based on the TADA likelihood are not consistent for the corresponding aDDM parameters. This further illustrates the importance of developing practical computational tools for the complex GDDMs of interest in cognitive neuroscience and other fields.

Our focus here is fast likelihood evaluation. For statistical inference it is often useful to also have likelihood gradients. Since the quadrature points used in our implementation of Algorithm \ref{alg: seq-multi} are fixed, the algorithm is inherently end-to-end differentiable by automatic differentiation software. In future work we hope to implement Algorithm \ref{alg: seq-multi} in an automatic differentiation package such as JAX \citep{jax2018github} to allow the use of gradient-based statistical inference methods, such as gradient descent for maximum likelihood estimation and Hamiltonian Monte Carlo for Bayesian inference. Strategies like GPU parallelization and parallel numerical integration could be used to gain even further speed-ups, which might be important in practice when using MCMC for inference in large hierarchical Bayesian models with GDDM likelihoods.

Our approach for multi-stage approximation should be generalizable to broader classes of SDEs than the ones considered here if they admit analytic formulas for FPTDs and NPDs in simple cases and if they have the appropriate Markov structure. If analytic formulas cannot be found, the multi-stage approximation might still be useful if accurate approximations of single-stage FPTDs and NPDs can be constructed by other methods, such as SBI.

There are several important mathematical questions that we leave for future work relating to regularity conditions for the existence of FPTDs, regularity conditions for the convergence of approximate multi-stage FPTDs under increasingly fine partitioning, and accuracy of numerical approximations for infinite series truncation, piecewise linear approximation of boundaries, and numerical quadrature. Despite these unresolved issues we believe our algorithms will find immediate use in the psychological and cognitive sciences and will improve the state-of-the-art for statistical inference of GDDMs.


\section*{Declarations}
\bmhead{Acknowledgements}
This work was supported by the National Institute of Mental Health (grants P50 MH119467-01 and P50 MH106435-06A1), the Office of Naval Research (MURI Award N00014-23-1-2792), and the Brainstorm Program at the Robert J. and Nancy D. Carney Institute for Brain Science. S.L. was additionally supported by the JBB Endowed Graduate Fellowship in Brain Science.

\bmhead{Contributions}
M.J.F. and A.F. raised the research question and provided expertise on the computational cognitive science applications. M.T.H. and S.L. developed the methodological and theoretical framework. S.L. carried out the research, validated the results, and prepared the associated GitHub repository with support from all authors. S.L. wrote the original draft with feedback and revisions from all authors. M.J.F. secured funding for the project. M.T.H supervised the project.

\bmhead{Data Availability}

The code supporting this study will be released as a package \textsc{EFPT} and made publicly available at \href{https://github.com/RiverFlowsInYou98/efficient-fpt}{https://\allowbreak github.com/\allowbreak RiverFlowsInYou98/efficient-fpt} upon publication.

\bmhead{Competing interests}

The authors declare no competing interests.

\begin{appendices}
\section{}

\subsection{Proof of Almost Sure Convergence of First Passage Times under Refined Boundary Approximations}\label{sec: approx_proof}
In this section, we establish a continuity result, which serves as the mathematical foundation of approximating the first passage times on curved boundaries by first passage times on their piecewise linear approximations, as described in Section \ref{sec: approx}. We first consider the simplified setting of a Brownian motion hitting a single upper boundary.

\begin{theorem}\label{thm: approx}
Let $W(t)$ be a one-dimensional Brownian motion on a probability space $(\Omega, \mathcal{F},\mathbb{P})$, and let $u(t)$ be a function on $[0,\infty)$ that is H\"{o}lder continuous with exponent $1/2$, satisfying $u(0)>0$. Suppose that a sequence of continuous functions $(u_n)$ converges uniformly to $u$ on compact subsets of $[0, \infty)$. Define the first passage times $\tau, \tau_n$ as
\begin{equation*}
\begin{aligned}
\tau&=\inf \{t>0: W(t) \ge u(t)\}\\
\tau_n&=\inf \{t>0: W(t) \ge u_n(t)\}\\
\end{aligned}
\end{equation*}
Then, $\tau_n$ converges to $\tau$ almost surely.
\end{theorem}
\begin{proof}[\proofname \text{ of Theorem \ref{thm: approx}}]
Define
\begin{equation*}
\begin{aligned}
\tau^{*}&=\inf \{t>0: W(t)> u(t)\}\\
\tau^{*}_n&=\inf \{t>0: W(t)> u_n(t)\}\\
\end{aligned}
\end{equation*}
We prove the desired result in three steps. We first prove that $\liminf_{n\rightarrow\infty} \tau_n \ge \tau$. Fix $\omega\in\Omega$. For any $\epsilon>0$, the function $t\mapsto u(t)-W(t)$ is continuous and attains its minimum $\delta=\delta(\epsilon)>0$ on $[0, \tau-\epsilon]$, i.e., $W(t)< u(t)-\delta$ when $t\le \tau-\epsilon$. Since $u_n \rightarrow u$ uniformly on $[0, \tau-\epsilon]$, there exists $N=N(\delta)=N(\epsilon)$ such that for all $n \geq N$ we have $u_n(t)>u(t)-\delta>W(t)$ for $t<\tau-\epsilon$. Thus, we have that for all $t<\tau-\epsilon$ and $n \geq N$, $W(t)<u_n(t)$. This implies $\tau_n \geq \tau-\epsilon$ for $n\ge N$. Taking $\epsilon \rightarrow 0$ gives $\liminf _{n \rightarrow \infty} \tau_n \geq \tau$.

Next, we prove that $\limsup_{n\rightarrow\infty} \tau^{*}_n \le \tau^{*}$. Fix $\omega\in\Omega$. By the definition of $\tau^{*}$, we know that for any $\epsilon>0$, there exists $t_\epsilon\in(\tau^{*}, \tau^{*}+\epsilon)$ and $\delta>0$, such that $W(t_\epsilon)>u(t_\epsilon)+\delta$. Since $u_n(t_\epsilon) \rightarrow u(t_\epsilon)$, there exists $N=N(\delta)=N(\epsilon)$ such that for all $n \geq N$ we have $u_n(t_\epsilon)<u(t_\epsilon)+\delta< W(t_\epsilon)$. Thus, there exists a $t_\epsilon\in (\tau^{*}, \tau^{*}+\epsilon)$ such that $u_n(t_\epsilon)<W(t_\epsilon)$ for $n\ge N$. This implies $\tau^{*}_n\le\tau^{*}+\epsilon$ for $n\ge N$. Taking $\epsilon\rightarrow 0$ gives $\limsup _{n \rightarrow \infty} \tau^{*}_n \leq \tau^{*}$. Hence $\limsup _{n \rightarrow \infty} \tau^{*}_n \leq \tau^{*}$.

Finally, we prove that $\tau=\tau^{*}$ a.s.. The inequality $\tau \leq \tau^*$ is immediate, so it suffices to show that $\mathbb{P}\left(\tau<\tau^*\right)=0$. Without loss of generality, we assume that $\tau$ is almost surely finite, since for any $\omega$ where $\tau(\omega)=\infty$, we trivially have $\tau(\omega)=\tau^*(\omega)=\infty$.

Take $\omega$ where $\tau(\omega)<\tau^{*}(\omega)$, then there exist $\delta>0$ such that $W(t)\le u(t)$ for $t\in (\tau, \tau+\delta)$. By the strong Markov property of $W$, we know that $B(s)= W(s+\tau)-W(\tau)$ is also a Brownian motion. Let $C$ be the H\"{o}lder $C^{0, 1/2}$ semi-norm of $u$, then we have
$u(t)\le u(\tau)+C(t-\tau)^{1/2}$ for $t>\tau$. Hence, we have
\begin{equation*}
W(\tau)+B(s)=W(s+\tau)\le u(s+\tau)\le u(\tau)+C s^{1/2}
\end{equation*}
when $s\in (0, \delta)$. By the continuity of $W$ and $u$ we have $W(\tau)=u(\tau)$, so $B(s)\le Cs^{1/2}$ for $s\in (0,\delta)$.

However, this would imply
\begin{equation}\label{eqn: lil-contradict}
\begin{aligned}
&\limsup _{s \rightarrow 0^+} \frac{B(s)}{\sqrt{2 s \log \log (1 / s)}}\le \limsup _{s \rightarrow 0^+} \frac{C}{\sqrt{2 \log \log (1 / s)}}=0
\end{aligned}
\end{equation}
Yet, the law of the iterated logarithm states that\begin{equation*}
 \limsup _{s \rightarrow 0^+} \frac{B(s)}{\sqrt{2 s \log \log (1 / s)}}=1\quad  \text{a.s.}  
\end{equation*}
so equation \eqref{eqn: lil-contradict} holds with probability $0$. So $\mathbb{P}(\tau< \tau^{*})=0$. 

Combining the above three steps, we obtain
\begin{equation*}
\tau\le \liminf_{n\rightarrow\infty} \tau_n\le \limsup_{n\rightarrow\infty} \tau_n\le\limsup_{n\rightarrow\infty}\tau_n^* \le\tau^*
\end{equation*}
from $\tau=\tau^*$ a.s., we conclude that $\lim_{n\rightarrow\infty} \tau_n=\tau$ a.s..
\end{proof}

Returning to the two-boundary setting, and adopting the notation in Section \ref{sec: approx}, the piecewise linear interpolations of the H\"{o}lder-$1/2$ continuous boundaries $u$ and $\ell$, denoted by $\overline{u}$ and $\overline{\ell}$ respectively, converge to $u$ and $\ell$ uniformly on compact subsets of $[0, \infty)$ as the interpolation grid is refined (see p35 of \cite{de1978practical}). Consequently, let $\bar{\tau}_u$ and $\bar{\tau}_{\ell}$ denote the approximate one-sided first passage times, and $\tau_u$ and $\tau_{\ell}$ their exact counterparts, then we can invoke Theorem \ref{thm: approx} to conclude that the $\overline{\tau}_u$ and $\overline{\tau}_\ell$ converge almost surely to $\tau_u$ and $\tau_\ell$ on their respective boundaries. Since $\overline{\tau}=\overline{\tau}_u\wedge\overline{\tau}_\ell, \tau=\tau_u\wedge\tau_\ell$, it follows that $\overline{\tau}\rightarrow\tau$ a.s.

\begin{figure*}[htb!]
  \centering
  \includegraphics[width=0.8\textwidth]{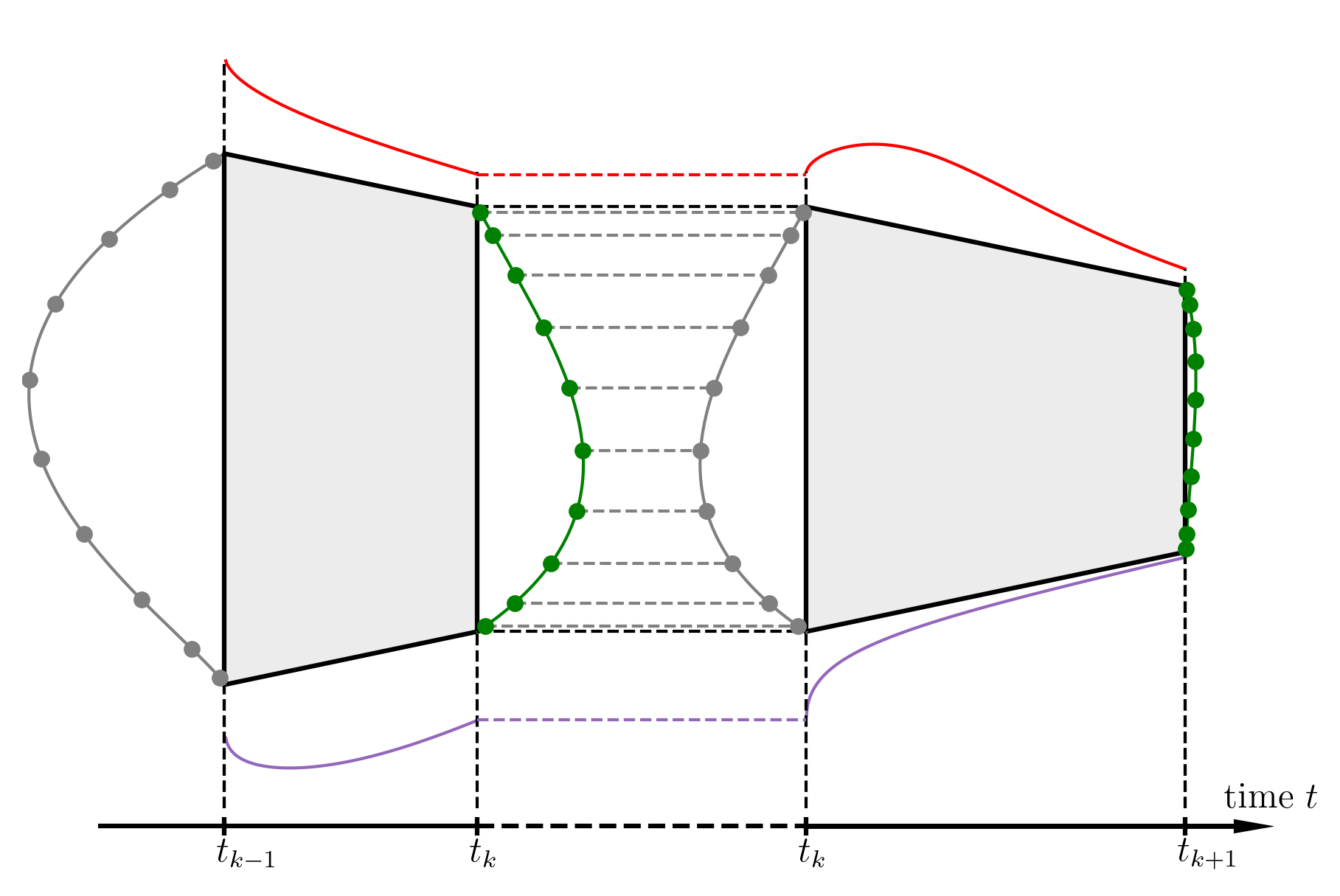}
\caption{At each vertical boundary between two stages, probability transfer occurs through a set of precomputed quadrature points—green dots on the output boundaries and grey dots on the input boundaries. The function $q_k$, evaluated at the quadrature points along the vertical boundary at $t=t_k$
(denoted as $\mathbf{q}_k$ in Algorithm \ref{alg: seq-multi-quad}), encapsulates the necessary information for integrating over the probability measure of the starting position $X(t_k)$, to facilitate computations in the next stage.}
\label{fig: seq_stitch}
\end{figure*}
\subsection{Transformation of a Diffusion Process to a Brownian Motion}\label{appendix: cherkasov}
Here, we present the Cherkasov condition as outlined in \cite{ricciardi1976transformation}, which provides the necessary and sufficient conditions for the existence of the transformation \eqref{eqn: transform} that converts a diffusion process $X(t)$ into a Brownian motion.
\begin{theorem}[\cite{ricciardi1976transformation}]
Let $X(t)$ be a diffusion process defined on the interval $[0, T_{\text{end}}]$ satisfying the $\operatorname{SDE}$ \eqref{eqn: sde}. If there exists a pair of functions $c_1(t)$ and $c_2(t)$ such that
\begin{equation}\label{eqn: cherkasov}
\begin{aligned}
&\mu(x, t)=\frac{\frac{\partial}{\partial x}\left(\sigma^2(x, t)\right)}{4}+\frac{\sigma(x, t)}{2}\Big(c_1(t)+\int_0^x \frac{c_2(t) \sigma^2(y, t)+\frac{\partial}{\partial t}\left(\sigma^2(y, t)\right)}{\sigma^3(y, t)} \diff y\Big)
\end{aligned}
\end{equation}
then there exists a coordinate transformation of the form \eqref{eqn: transform} such that $\psi(X(\gamma^{-1}(s)), \gamma^{-1}(s))$, the image of $X(t)$ under this transformation, has the same dynamics with a standard Brownian motion. This transformation is given by
\begin{equation*}
\begin{aligned}
s=\gamma(t)=&\int_{0}^t e^{-\int_{0}^r c_2(z) \diff z} \diff r\\
w=\psi(x, t)=&e^{-\frac{1}{2} \int_{0}^t c_2(s) \diff s} \int_0^x \frac{1}{\sigma(y, t)}\diff y-\frac{1}{2} \int_{0}^t c_1(r) e^{-\frac{1}{2} \int_{0}^r c_2(z) \diff z} \diff r\\
\end{aligned}
\end{equation*}
\end{theorem}

We take the Ornstein-Uhlenbeck process \eqref{eqn: ou} as an example. Here $\mu(x, t)=\theta(\lambda-x)$ and $\sigma(x,t)=\sigma$, and the Cherkasov condition \eqref{eqn: cherkasov} reads
\begin{equation*}
\theta(\lambda- x)=\frac{\sigma}{2} c_1(t)+\frac{x}{2} c_2(t)
\end{equation*}
We can take $c_1(t)=\frac{2\theta\lambda}{\sigma}, c_2(t)=-2\theta$. Correspondingly, the transformation is given by $s=\gamma(t)=\frac{e^{2\theta t}-1}{2\theta}, w=\psi(x, t)=\frac{1}{\sigma}(e^{\theta t} x-\lambda e^{\theta t}+\lambda)$. This is the same as the transformation given in Section \ref{sec: generalization}. Verifying the Cherkasov condition for general diffusion processes requires analyzing equation \eqref{eqn: cherkasov} on a case-by-case basis.

\subsection{Fast Implementation}\label{sec: implementation}
In this section, we discuss some practical aspects of our algorithm that contribute to its fast execution, including important details about numerical integration and function representation that are glossed over in the description of Algorithm \ref{alg: seq-multi}. We present Algorithm \ref{alg: seq-multi-quad}, the practical version of Algorithm \ref{alg: seq-multi} to conclude this section.

\begin{itemize}
	\item If $t\in (t_{k-1},t_k]$, then we can halt Algorithm \ref{alg: seq-multi} after the evaluation of $f_c(t)$ in the $k$-th stage. Also, note that evaluations of $f_u$ and $f_\ell$ can be skipped unless they are needed for the final output, since only $q_k$ is propagated across stages.
	
	\item Formulas \eqref{eqn: fptd_upper-basic}, \eqref{eqn: fptd_lower-basic} and \eqref{eqn: fptd_vertical-basic} are approximated by truncating the infinite sum to a desired accuracy. The truncation criteria can be used as a way to balance speed and accuracy. 
	\item Piecewise linear approximations of the transformed boundaries in Section \ref{sec: generalization} are realized via adaptive linear interpolation, which balances accuracy and computational efficiency by dynamically adjusting both the number and position of interpolation points based on the function's variability. 
    
    More specifically, we adopt the following procedure in our numerical examples (Sections \ref{sec: example2}, \ref{sec: example3} and \ref{sec: error_example}): we initialize a knot set containing all stage endpoints in the original model, construct a piecewise linear interpolant through the corresponding boundary values, and iteratively refine the knot set by inserting additional knots at locations where the current interpolant attains its largest absolute deviation from the true boundary on a fine evaluation grid. The refinement terminates once this maximum deviation falls below a prescribed tolerance $\varepsilon$ on every subinterval. The tolerance $\varepsilon$ controls the fineness of the resulting interpolation grid.

\end{itemize}

To compute the likelihood given a dataset, we simply iterate Algorithm \ref{alg: seq-multi} over all data. Additional dataset-level speedups can be achieved using the following strategies:
\begin{itemize}
\item When multiple observations share {\em the same} GDDM configuration, a single pass of Algorithm \ref{alg: seq-multi} suffices to evaluate their likelihoods simultaneously, eliminating redundant computations. This idea underlies the speedups achieved by many previous methods (e.g., KFE-based approaches). However, such methods slow down considerably when GDDMs vary across trials, whereas our approach adapts seamlessly to both homogeneous and heterogeneous settings while maintaining high computational efficiency.

\item The computation of each configuration in a dataset is independent of each other, rendering the problem ``embarrassingly parallelizable" across configurations --- we can distribute the dataset over multiple threads, and since no communication is needed between the sub-tasks, the computational time should ideally be inversely proportional to the number of threads, up to the point where the number of threads equals the number of configurations. Consequently, models such as the aDDM, where each trial has a distinct configuration, can benefit substantially from this form of parallelization. We repeat the timing experiment for Algorithm \ref{alg: seq-multi} described in Section \ref{sec: timing} with multiple threads and report the results in Table \ref{tab: aDDM_parallel}, which demonstrate that our method can easily be further accelerated through parallelization.

\end{itemize}

\begin{table}[htb!]
	\centering
    \caption{Evaluation times for the aDDM likelihood \eqref{eqn: likeADDM} with $50{,}000$ trial data, using Algorithm \ref{alg: seq-multi} under different thread counts.} 
\begin{tabular*}{\columnwidth}{@{\hspace{\tabcolsep}\extracolsep{\fill}}ccccccc@{\hspace{\tabcolsep}}}
\toprule
\textbf{\# Threads} & 1 & 2 & 4 & 8 & 16 & 32\\\midrule
\textbf{Time$(\mathrm{s})$}&5.547 &2.698 &1.369& 0.760& 0.424& 0.264\\\bottomrule
\end{tabular*}
\label{tab: aDDM_parallel}
\end{table}

\subsubsection{Detailed Implementation of the Quadrature}\label{appendix: quad}

\begin{algorithm*}[htb!]
\SetAlgoRefName{1$'$}
\SetKwInOut{Input}{Input}
\SetKwInOut{Output}{Output}
\caption{Sequential likelihood evaluation of a multi-stage GDDM, with Gauss-Legendre quadrature}
\label{alg: seq-multi-quad}
\Input{A multi-stage model configuration as described in Section \ref{sec: multi}; An observation $(t,c)$ for $t_{K-1}<t\le t_K$ ($t$ is in the $K$-th stage) and $c\in\{u, \ell\}$ or an observation of ``non-response", in the latter case the number of effective stages $K$ is set to be $d+1$; $q_0$ as the sub-probability ``density" of $X(0)$.}
\Output{$f(t,c)=f_c(t)$ for a datum pair $(t,c)$, or $Q$ for ``non-response"}
\BlankLine
\uIf(\tcc*[h]{$X(0)$ is continuous}){$q_0\text{ is a pdf}$}{$\mathbf{w}_0, \boldsymbol{\xi}_0\gets\text{gausslegendre}(m_0)$\;
$\mathbf{x}_{0}\gets\frac{u(t_0)-\ell(t_0)}{2}\boldsymbol{\xi}_0+\frac{\ell(t_0)+u(t_0)}{2}$\;
$\mathbf{q}_0=q_0(\mathbf{x}_0)$\;}
\ElseIf(\tcc*[h]{$X(0)$ is discrete}){$q_0=\sum_{j=1}^J w_{0j}\delta_{x_{0j}}$}{
$\mathbf{w}_0\gets (w_{01}, \cdots, w_{0J})^\top$\;
$\mathbf{x}_0\gets(x_{01}, \cdots, x_{0J})^\top$\;
$\mathbf{q}_0\gets(1, \cdots, 1)^\top$\;
}

\For{$k = 1$ \KwTo $K-1$}{
$\mathbf{w}_k, \boldsymbol{\xi}_k\gets\text{gausslegendre}(m_{k})$\;
$\mathbf{x}_{k}\gets\frac{u(t_k)-\ell(t_k)}{2}\boldsymbol{\xi}_k+\frac{\ell(t_k)+u(t_k)}{2}$\;
$\mathbf{P}_k\gets q^{\text{single}}(\mathbf{x}_k; \mu_k, \sigma_k, \mathcal{B}_k, \mathbf{x}_{k-1}^\top)$\;
$\mathbf{q}_k\gets\mathbf{P}_k(\mathbf{q}_{k-1}\odot \mathbf{w}_{k-1})$\;
}
\uIf{observation is $(t,c)$}{$\mathbf{c}_K\gets f_c^{\text{single}}(t-t_{K-1};\mu_K, \sigma_K, \mathcal{B}_K, \mathbf{x}_{K-1}^\top)$\;
$f(t, c)\gets\mathbf{c}_K(\mathbf{q}_{K-1}\odot \mathbf{w}_{K-1})$\;}
\ElseIf{observation is ``non-response"}{$Q\gets \mathbf{w}_d^\top\mathbf{q}_d$}
\end{algorithm*}

The $m$-point Gauss-Legendre quadrature is given by
\begin{equation}\label{eqn: GL-quad}
\begin{aligned}
\int_a^b f(x) \diff x \approx&\frac{b-a}{2} \sum_{i=1}^m w_i f\Big(\frac{b-a}{2} \xi_i+\frac{a+b}{2}\Big)=\frac{b-a}{2}\mathbf{w}^\top f\Big(\frac{b-a}{2}\boldsymbol{\xi}+\frac{a+b}{2}\Big)\\
\end{aligned}
\end{equation}
where $\mathbf{w}=(w_1,\cdots,w_m)^\top$ are the quadrature weights, $\boldsymbol{\xi}=(\xi_1,\cdots, \xi_m)^\top$ are the quadrature points on the reference interval $[-1,1]$. The formula \eqref{eqn: GL-quad} is exact for polynomials of degree up to $2m-1$. $\mathbf{w}$ and $\boldsymbol{\xi}$ can be computed by formulas involving Legendre polynomials but are usually tabulated to great accuracy.

To apply formula \eqref{eqn: GL-quad} to our algorithm, at each stage $k$, we only need to track the density at precomputed quadrature points on each vertical boundary. We then proceed to the next stage by computing the transitional densities $q_k$ from the current quadrature points to the quadrature points on the next vertical boundary (see \figref{fig: seq_stitch}). Regarding the order of quadrature, the functions in equations \eqref{eqn: fptd_upper-basic} and \eqref{eqn: fptd_lower-basic} are observed to exhibit a more rapid growth of derivatives compared to those in equation \eqref{eqn: fptd_vertical-basic}, hence the integrand in the final stage is less regular than in the earlier stages. Therefore, we typically employ a quadrature rule of equal or higher order in the final stage. In other words, if we let $m_k$ denote the quadrature order for integration in the $k$-th stage, then when $\tau$ lies in the $K$-th stage, we can choose $m_K \geq m_i$ for $i=1, \ldots, K-1$ to achieve more efficient computation.

Here, we present Algorithm \ref{alg: seq-multi-quad}, which incorporates numerical integration via \eqref{eqn: GL-quad} into Algorithm \ref{alg: seq-multi}. Note that different initial treatments are needed for discrete and continuous $X(0)$, as if $X(0)$ is discrete, we only need to directly use formulas \eqref{eqn: fptd_upper}, \eqref{eqn: fptd_lower} and \eqref{eqn: fptd_vertical}, and there is no need for integration in the first stage. In cases where \( X(0) \) is a mixture of discrete and continuous components, we handle each part accordingly. However, we omit this case in Algorithm \ref{alg: seq-multi-quad} for brevity, as its treatment follows naturally from the discrete and continuous cases. 

In Algorithm \ref{alg: seq-multi-quad}, we denote the abstract process of acquiring $\mathbf{w}, \boldsymbol{\xi}$ of order $m$ as $\mathbf{w}, \boldsymbol{\xi}\leftarrow\text{gausslegendre}(m)$, and the symbol \( \odot \) represents the Hadamard (element-wise) product for vectors. Assuming constant time complexity for evaluating the single-stage FPTDs $f_u^{\text{single}}, f_\ell^{\text{single}}$ and NPD $q^{\text{single}}$, the overall time complexity of Algorithm \ref{alg: seq-multi-quad} is $\Theta(\sum_{k=1}^K m_{k-1}m_k)$.

\subsection{Error Analyses of Algorithm \ref{alg: seq-multi-quad}}

There are three sources of error that contribute to the overall approximation error of the FPTDs and NPDs produced by Algorithm \ref{alg: seq-multi-quad}:
\begin{itemize}
    \item Truncating the infinite sums in the single-stage formulas \eqref{eqn: fptd_upper-basic}, \eqref{eqn: fptd_lower-basic} and \eqref{eqn: fptd_vertical-basic} introduces a controllable truncation error whose magnitude depends on the chosen truncation criterion. Empirically, we observe that the series usually converge numerically with very few terms for $t$ that is not too large.  
    \item Algorithm \ref{alg: seq-multi} evaluates several integrals when propagating NPD across stages and when marginalizing over random starting positions. Approximating these integrals via Gauss-Legendre quadrature \eqref{eqn: GL-quad} induces a numerical integration error that decreases as the quadrature order $m_k$ is increased.
    \item For general GDDMs with non-piecewise-linear boundaries, the approximation of the original boundaries by their piecewise linear interpolants (Section \ref{sec: approx}) on a time partition introduces an approximation error that goes down as the partition is refined.
\end{itemize}
The truncation criterion for the infinite series, the quadrature order, and the fineness of the time partition used to interpolate the boundaries can all be tuned to balance computational cost and accuracy. Empirically, we observe that for multi-stage models, using a moderate number of series terms (about $10$) together with a moderate Gauss–Legendre quadrature order (about $20–30$), Algorithm \ref{alg: seq-multi-quad} already achieves a high accuracy that is more than sufficient for practical purposes. For GDDMs with non-piecewise-linear boundaries, the overall error is typically dominated by the boundary approximations.
\subsubsection{A Numerical Study of Errors}\label{sec: error_example}
Here we illustrate, through a simple example, how boundary approximation and numerical integration impacts the resulting FPTD. We consider a standard Brownian motion 
\begin{equation*}
X(t)=x_0+W(t)    
\end{equation*}
starting at $x=x_0$, 
and study its first passage time to the two square root boundaries
\begin{equation*}
u(t)=a\sqrt{T-t},\quad \ell(t)=-b\sqrt{T-t},
\end{equation*}
where $a, b, T>0, x_0\in (-b\sqrt{T}, a\sqrt{T})$.

To obtain reference values for the FPTDs, we apply the time change $s=\log(T/(T-t))$, which transforms the original problem into a first passage time problem for the (unstable) Ornstein–Uhlenbeck process
\begin{equation*}
\diff \widetilde{X}(s)=\frac{1}{2} \widetilde{X}(s) \diff s+\diff  \widetilde{W}(s),
\end{equation*}
where 
\begin{equation*}
\widetilde{W}(s)=\int_0^{T(1-e^{-s})} \frac{1}{\sqrt{T-r}} \diff W(r)
\end{equation*}
is another standard Brownian motion. Under this transformation, the two absorbing boundaries become constant:
\begin{equation*}
\widetilde{u}(s)=a,\quad \widetilde{\ell}(s)=b.
\end{equation*}
The reference FPTDs are then obtained by solving the KFE governing the transformed problem using a highly accurate spectral element method.

In this example, we take $a=0.6, b=0.4, T=5, x_0=0.1$. Following the procedure in Section \ref{sec: summary}, we approximate the boundaries using adaptive piecewise linear interpolation with tolerance $\varepsilon$. The tolerance $\varepsilon$ controls the fineness of the resulting approximation; we then compute the corresponding approximate FPTDs using Algorithm \ref{alg: seq-multi-quad} at different quadrature orders. We compare these outputs to the reference FPTDs, and the results are displayed in Figure \ref{fig: error_llhd} and Table \ref{tab: error}. 
\begin{figure*}[htb!]
	\centering
	\includegraphics[width=0.8\textwidth]{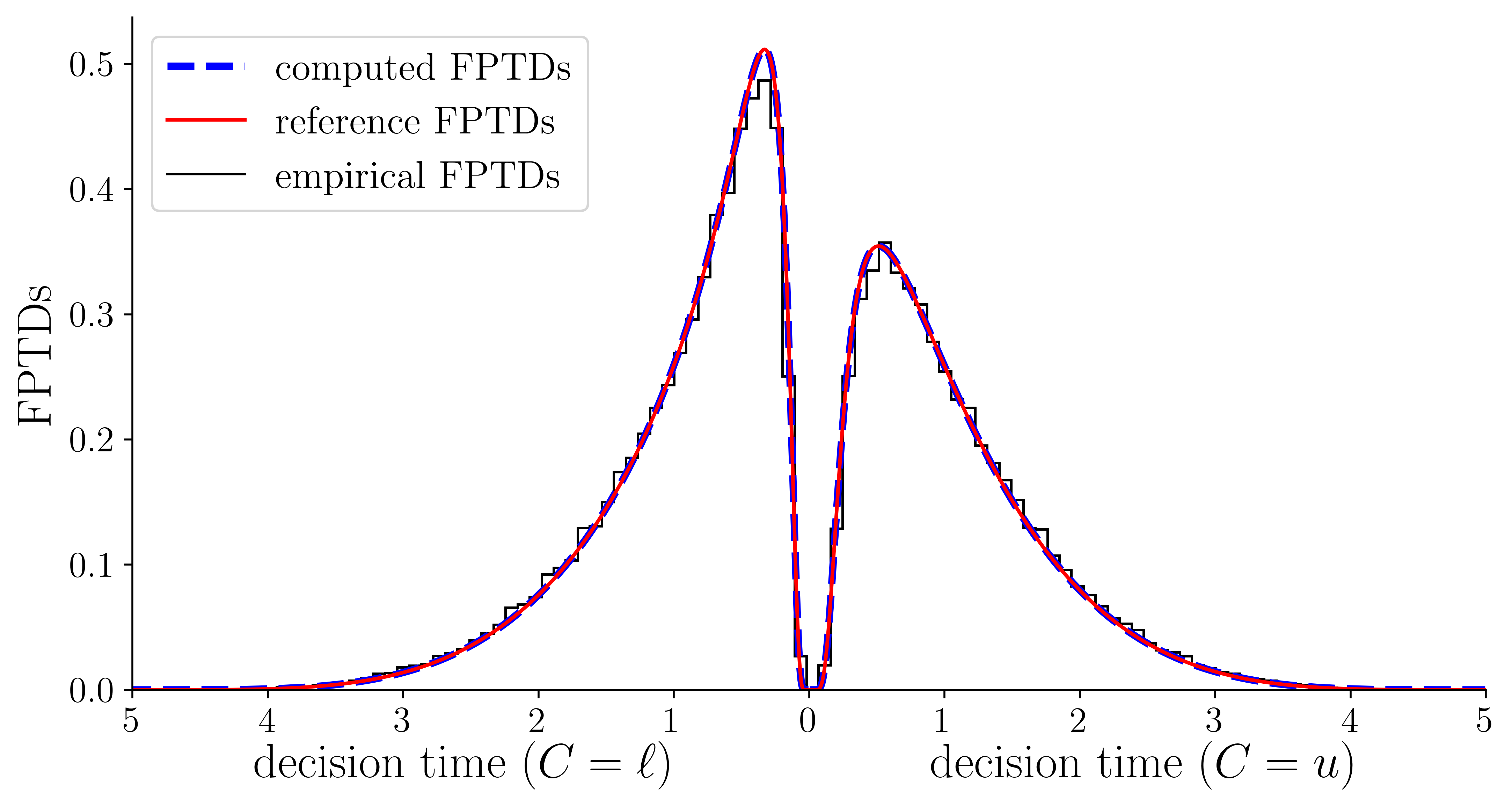}
	\caption{Plots of FPTDs and histogram for a standard Brownian motion hitting two square root boundaries. First passage times on the upper boundary are displayed on the positive $t$-axis and those on the lower boundary are displayed on the negative $t$-axis so that the whole plot is a valid probability density function. The blue dashed line shows the FPTDs produced by Algorithm \ref{alg: seq-multi-quad} with piecewise linear interpolation tolerance $10^{-4}$ and quadrature order $m_k=20$ for all $k$, the red solid line shows the reference FPTDs obtained from solving KFE to very high accuracy, and the histogram is estimated from simulated first passage time data generated using Euler-Maruyama scheme \eqref{eqn: euler-maruyama} with step size $10^{-5}$.}
	\label{fig: error_llhd}
\end{figure*}

\begin{table}[htb!]
	\centering
    
    \caption{Errors of the FPTDs produced by Algorithm \ref{alg: seq-multi-quad} measured relative to the reference FPTDs. The reported error is the total variation distance between the two density functions, computed over a uniform grid of $1000$ time points on $[0,T]$.
    }
    \label{tab: error}

\begin{tabular}{cccccc}
\toprule
\textbf{Tolerance $\varepsilon$}  & $10^{-2}$ & $10^{-3}$ & $10^{-4}$ & $10^{-5}$ & $10^{-6}$ \\\midrule
\textbf{\# Stages} & 11 & 32 & 89 & 259 & 577 \\\midrule
\textbf{TV distance} &$1.849\times 10^{-3}$ &$5.297\times10^{-4}$& $5.292\times10^{-5}$  & $7.113\times10^{-6}$& $2.397\times 10^{-6}$\\\bottomrule
\end{tabular}
\subcaption{We fix the quadrature order $m_k=30$ and consider a sequence of decreasing interpolation tolerances $\varepsilon$, which yields increasingly fine piecewise linear boundary approximations and thus multi-stage models with more stages. We then apply Algorithm \ref{alg: seq-multi-quad} to each approximation and report the resulting errors.}
\resizebox{\columnwidth}{!}{%
\begin{tabular}{cccccc}
\toprule
\textbf{Quadrature order $m_k$}  & $10$ & $15$ & $20$ & $25$ & $30$ \\\midrule
\textbf{TV distance} &$1.482\times 10^{-3}$ &$2.959\times 10^{-4}$& $6.446\times 10^{-5}$  & $5.319\times 10^{-5}$& $5.292\times 10^{-5}$\\\bottomrule
\end{tabular}
}
\subcaption{We approximate the boundaries using adaptive piecewise linear interpolation with $\varepsilon=10^{-4}$, and then apply Algorithm \ref{alg: seq-multi-quad} with a sequence of increasing quadrature orders $m_k$ (taken to be the same across all stages $k$ in this experiment). We report the resulting errors. As $m_k$ increases, the errors initially decrease and then plateau since boundary approximation error dominates.}
\end{table}

From Figure \ref{fig: error_llhd}, we see that the approximate FPTDs closely match the reference FPTDs and are consistent with the simulation results. From Table \ref{tab: error}, the total variation distance between the approximate FPTD and the reference FPTD function decreases monotonically as $\varepsilon$ is reduced and as the quadrature order increases, indicating that finer boundary approximations and higher quadrature order lead to systematically more accurate FPTDs.
\end{appendices}


\newpage
\bibliography{sn-bibliography}

\end{document}